\documentclass[11pt]{article}
\pdfoutput=1
\RequirePackage{algorithmic}
\RequirePackage{algorithm}

\algsetup{indent=2em}

\newcommand {\rNokgjXa} {\REQUIRE}
\newcommand {\BeqmjZec} {\ENSURE}

\newcommand{\UDODcOFW}[4][h t b]{
\begin{algorithm}[#1]
\caption{#2}
\label{alg:#3}
\begin{algorithmic}[1]
#4
\end{algorithmic}
\end{algorithm}
}
\newcommand {\EfCTdHOd} [1] { \quad \COMMENT{#1}}
\RequirePackage{color}
\RequirePackage{amssymb}
\RequirePackage{amsfonts}
\RequirePackage{bbm}
\RequirePackage{hyperref}
\RequirePackage{bookmark}
\RequirePackage{amsthm}
\RequirePackage{mathtools}
\newcommand {\hXDTkvgo} {
\RequirePackage[a4paper, twoside,
ignoreheadfoot, nohead, nomarginpar,
width=14.6cm, hcentering,
top=2.6cm, vmarginratio=3:4]{geometry}
\HAmZsFFB
\ZTAzoJRA
\CKZbVqUA
\vDplsxWT
\JjzwRmFi
\RequirePackage{flafter}
}
\newcommand {\HAmZsFFB} {
\RequirePackage[T1]{fontenc}
\RequirePackage{lmodern}
\RequirePackage{mathptmx}
\RequirePackage[scaled=.92]{helvet}
\RequirePackage{eucal}
}
\newcommand {\ZTAzoJRA} {
\RequirePackage{fancyhdr}
\pagestyle{fancyplain}
\fancyhf{}
\renewcommand{\headrulewidth}{0pt}
\renewcommand{\footrulewidth}{0pt}
\newcommand{\PJMUQtLt}{\textsl{\small\thepage}}
\fancyfoot[C]{\fancyplain{\PJMUQtLt}{}}
\fancyfoot[LE,RO]{\fancyplain{}{\PJMUQtLt}}
}
\newcommand{\vDplsxWT}{ \RequirePackage[stretch=10,shrink=10]{microtype} }
\newcommand{\JjzwRmFi}{ \RequirePackage[nodayofweek]{datetime} }
\newcommand {\LTNsnnmo} [1] {\href{mailto:#1}{\texttt{#1}}}
\newcommand {\GxamKrlf} {Attila Pereszl\'{e}nyi}
\newcommand {\INbSqQEv} {\LTNsnnmo{attila.pereszlenyi@gmail.com}}
\newcommand {\pSisqEJw}
{Centre for Quantum Technologies, National University of Singapore}
\newcommand {\LDYwFcAX} {without loss of generality}
\newcommand {\GmcWftaL} {such that}
\newcommand {\YRlgvUEY} {i.e.,}
\newcommand {\djIkZXTK} {e.g.,}
\newcommand {\TbXXHkNG} {Hilbert space}
\newcommand {\vHMHwcRb} {Bloch sphere}
\newcommand {\CTcThSvb} {Choi-Jamio{\l}kowski}
\newcommand {\dFUQloAO} {\CTcThSvb{} representation}
\newcommand {\ExDgacVq} {de Finetti}
\newcommand {\hFGJtSId} {et al.}
\newcommand {\wRzSJcrL} {quantum interactive proof system}
\newcommand {\WlSHGbnz} {
\author {\GxamKrlf\thanks{E-mail: \INbSqQEv.}\\
\textsl{\small \pSisqEJw}}
}
\newcommand {\iivqMZNr} [1] {
\WlSHGbnz
\title{\textbf{#1}}
\wcQJzGlS{#1}{\GxamKrlf}
}
\newcommand {\fERcsqXi} {
\bibliographystyle{halpha}
\bibliography{./bib}
}
\newcommand {\XsBWhkhM} [1] {Theorem~\ref{#1}}
\newcommand {\cEgZOrfG} [1] {Lemma~\ref{#1}}
\newcommand {\ozsaYFnP} [1] {Definition~\ref{#1}}
\newcommand {\nvBBVayi} [1] {Algorithm~\ref{#1}}
\newcommand {\BeRtWGMB} [1] {Section~\ref{#1}}
\newcommand {\GTIJeHWG} [1] {\ensuremath{ \left( #1 \right) }}
\newcommand {\YnJXqMTX} [1] {\ensuremath{ \left[ #1 \right] }}
\newcommand {\dVhhxuvI} [1] {\ensuremath{ \left\lbrace #1 \right\rbrace }}
\newcommand {\VslhEMHJ} {\;}
\newcommand {\GscNOKTs} {\: \! \!}
\newcommand {\mcpkmsYW} [2] {\ensuremath{ #1 \GscNOKTs \GTIJeHWG{ #2 } }}
\newcommand {\aJAxrBMB} {\ensuremath{ \stackrel {\mathrm{def}} {=} }}
\newcommand {\HFkfLWeq} {\ensuremath{ \stackrel{\mathrm{?}}{=} }}
\newcommand {\XDGqcIBa} {\ensuremath{ \iota }}
\newcommand {\WrVYYpex} [1] {\mcpkmsYW {O} {#1}}
\newcommand {\hdzgZeCW} [1] {\dVhhxuvI{ 1, 2, \dotsc, #1 }}
\newcommand {\lyfuwjUX} {\, : \,}
\newcommand {\LtfmKMob} [2] {\dVhhxuvI{ #1 \lyfuwjUX #2 }}
\newcommand {\BuwGPLRs} [3] {\ensuremath{ #1 \LtfmKMob{#2}{#3} }}
\newcommand {\rRGRugxs} [3] {\ensuremath{ #1 : \, #2 \rightarrow #3 }}
\newcommand {\PbCNIRqc} [1] {\ensuremath{ \mathbb{#1} }}
\newcommand {\ynpHhbyi} {\PbCNIRqc{C}}
\newcommand {\FLbfwYzZ} {\ensuremath{ \PbCNIRqc{Z}^{+} }}
\newcommand {\zmKbSHOE} {\ensuremath{ \in_{\mathrm{R}} }}
\newcommand {\bymhjSLm} [1] {\ensuremath{ \left\langle #1 \right| }}
\newcommand {\CARSeAIr} [1] {\ensuremath{ \left| #1 \right\rangle }}
\newcommand {\SCAQAFmU} [2] {\ensuremath{ \left| #1 \middle\rangle \middle\langle #2 \right| }}
\newcommand {\yAIHQVQC} [1] {\SCAQAFmU{#1}{#1}}
\newcommand {\mlpiBNcl} [2] {\ensuremath{ \left\langle #1 \middle\vert #2 \right\rangle }}

\newcommand {\FGGdfwUg} {\ensuremath{ \frac{\CARSeAIr{00} + \CARSeAIr{11}}{\sqrt{2}} }}
\newcommand {\UPMtlSUt} {\ensuremath{ \otimes }}
\newcommand {\XhMHOPJK} [1] {\ensuremath{ #1^{*} }}
\newcommand {\EovgbaUD} {\ensuremath{ \mathrm{Tr} }}
\newcommand {\gyBwVFHt} [1] {\mcpkmsYW{\EovgbaUD}{#1}}
\newcommand {\FbNPrGDa} [2] {\mcpkmsYW{ \EovgbaUD_{#1} }{ #2 }}
\newcommand {\UylNyihM} [1] {\ensuremath{ \left| #1 \right| }}
\newcommand {\eNNesbxu} [1] {\ensuremath{ \left\| #1 \right\| }}
\newcommand {\xotrFLoO} [2] {\ensuremath{ \eNNesbxu{#1}_{#2} }}
\newcommand {\JjPvCTWX} [1] {\xotrFLoO{#1}{\EovgbaUD}}
\newcommand {\bjhdjWcQ} [1] {\xotrFLoO{#1}{\infty}}
\newcommand {\DGaDmHJe} [2] {\mcpkmsYW{\mathrm{F}}{#1, #2}}
\newcommand {\HcxXjAcJ} [1] {\ensuremath{ \mathsf{#1} }}
\newcommand {\evJAjzTZ} [1] {\ensuremath{\text{\textsc{#1}}}}
\newcommand {\OYsGslWI} {\HcxXjAcJ{NP}}
\newcommand {\fGxhDjlY} {\HcxXjAcJ{PSPACE}}
\newcommand {\NEDasFOP} {\HcxXjAcJ{IP}}
\newcommand {\KRCqOJJJ} {\HcxXjAcJ{MA}}
\newcommand {\FEZMUXFw} {\HcxXjAcJ{PP}}
\newcommand {\wkGMBCsG} {\HcxXjAcJ{A_{0}PP}}
\newcommand {\ihRQQChH} {\HcxXjAcJ{QMA}}
\newcommand {\vPhVIqeI} {\HcxXjAcJ{QCMA}}
\newcommand {\GUYckNnU} {\HcxXjAcJ{QIP}}
\newcommand {\DYmPyZgR} {\HcxXjAcJ{BQP}}
\newcommand {\pbpWqYJX} [2] {\mcpkmsYW{\ihRQQChH}{#1, #2}}
\newcommand {\HnNIoaAl} {\ensuremath{ \KRCqOJJJ_{1} }}
\newcommand {\medaFzJc} {\ensuremath{ \ihRQQChH_{1} }}
\newcommand {\HvMKrrDT} {\ensuremath{ \vPhVIqeI_{1} }}
\newcommand {\qOrOiFzp} {\ensuremath{ \GUYckNnU_{1} }}
\newcommand {\jYpQwXZd} [1] {\ensuremath{ \mathnormal{#1} }}
\newcommand {\BWHXbaAC} [1] {\ensuremath{ \mathbf{#1} }}
\newcommand {\GECLKEJq} {\BWHXbaAC{CNOT}}
\newcommand {\eSvaqyjm} {\BWHXbaAC{H}}
\newcommand {\FqTFSrwh} {\BWHXbaAC{X}}
\newcommand {\RUYzfkFo} {\BWHXbaAC{Z}}
\newcommand {\WsoGCMAU} {\BWHXbaAC{Y}}
\newcommand {\cNXPyGJX} {\BWHXbaAC{T}}
\newcommand {\aOIUkcRb} {\ensuremath{ \mathbbm{1} }}
\newcommand {\FyfkKUDd} [1] {\ensuremath{ \mathsf{#1} }}
\newcommand {\nuoeaWUW} [1] {\ensuremath{ \mathcal{#1} }}
\newcommand {\FnwnljbK} [1] {\ensuremath{\mathrm{#1}}}
\newcommand {\poaESbLo} [1] {\mcpkmsYW{\FnwnljbK{L}}{#1}}
\newcommand {\sYllRVEw} [1] {\mcpkmsYW{\FnwnljbK{D}}{#1}}
\newcommand {\CKZbVqUA} {
\theoremstyle {plain}
\newtheorem {theorem} {Theorem} [section]
\newtheorem {corollary} [theorem] {Corollary}
\newtheorem {lemma} [theorem] {Lemma}
\theoremstyle {remark}

\theoremstyle {definition}
\newtheorem {definition} [theorem] {Definition}

}
\newcommand {\GzyBBfAC} {\tag*{\qedhere}}
\newcommand {\wcQJzGlS} [2] {
\definecolor{color_of_url}{RGB}{0,0,75}
\hypersetup{
pdftitle={#1},
pdfauthor={#2},
pdfstartview={FitH},
breaklinks={true},
colorlinks={true},
linkcolor={black},
citecolor={black},
urlcolor={color_of_url}
}
}
\hXDTkvgo
\newcommand {\gAuWdsUU} {\ensuremath{N}}
\newcommand {\IwFBcXWd} [1] {\ensuremath{ \BWHXbaAC{W}_{#1} }}
\newcommand {\bYQpVZwG} {\ensuremath{ \ihRQQChH^{\text{\textnormal{const-EPR}}} }}
\newcommand {\hxlhxXVG} {\ensuremath{ \medaFzJc^{\text{\textnormal{const-EPR}}} }}
\newcommand {\mjIKlHwf} [2] {\mcpkmsYW{\bYQpVZwG}{#1, #2}}
\newcommand {\qIsFHsNm} [1] {\ensuremath{ \CARSeAIr{\mcpkmsYW{J}{#1}} }}
\newcommand {\IbAozRcX} {\ensuremath{ \aOIUkcRb_{\nuoeaWUW{P}} \UPMtlSUt \KMFbyGRX }}
\newcommand {\ulbwWGDE} {\ensuremath{ \Pi_{\mathrm{acc}} }}
\newcommand {\CGvPMhCL} {\ensuremath{ \widetilde{\Pi}_{\mathrm{acc}} }}
\newcommand {\lmkGSgYf} {\ensuremath{ \BWHXbaAC{M}_x }}
\newcommand {\LZfegDul} {\ensuremath{ \BWHXbaAC{V}_x }}
\newcommand {\hhqeGUPm} {Distillation Procedure}
\newcommand {\KlrOGnZA} {SWAP Test}
\newcommand {\LkIoSNVy} {Reflection Simulation Test}
\newcommand {\gorjOgEG} {Reflection Procedure}
\newcommand {\mCDuGkWD} {Space Restriction Test}
\newcommand {\Vohxuxyh}[2]{\mcpkmsYW{d}{#1, #2}}
\newcommand {\LBGptwLj} {\ensuremath{\bar{0}}}
\newcommand {\MEVIdGJb} {\ensuremath{\phi^{+}}}
\newcommand {\VPzthEpq} {\ensuremath{\phi^{-}}}
\newcommand {\FZGQkLzf} {\cNXPyGJX \sigma_i \XhMHOPJK{\cNXPyGJX}}
\newcommand {\tVRtKpaI} {\ensuremath{\mu_{\fTIaZVxF}}}
\newcommand {\RWgRdDui} {\ensuremath{\nu_{\fTIaZVxF}}}
\newcommand {\MrWNkHpB} {\ensuremath{\frac{\aOIUkcRb_{\nuoeaWUW{S}_1'
\UPMtlSUt \nuoeaWUW{S}_2'}}{4}}}
\newcommand {\KMFbyGRX} {\ensuremath{\yAIHQVQC{\LBGptwLj}_{\nuoeaWUW{A}}}}
\newcommand {\fTLcMlGS} [2] {\ensuremath{\yAIHQVQC{#1}_{\nuoeaWUW{S}_{#2}}}}
\newcommand {\dIKLUvxq} {\GTIJeHWG{\aOIUkcRb - \CGvPMhCL}}
\newcommand {\rkQwwcQs} {\ensuremath{\UuexCuFw^{+}}}
\newcommand {\tTkICuaj} {\ensuremath{\UuexCuFw^{-}}}
\newcommand {\UuexCuFw} {\ensuremath{A}}
\newcommand {\fTIaZVxF} {\ensuremath{B}}
\newcommand {\CQiPcAHs} {\FbNPrGDa{\nuoeaWUW{S}_2}{\sigma_i}}
\newcommand {\LvsCZiIr} {\ensuremath{\xi}}
\newcommand {\hxCscLfe} {p_{+} \GTIJeHWG{\yAIHQVQC{\MEVIdGJb}}^{\UPMtlSUt 2}
+ p_{-} \GTIJeHWG{\yAIHQVQC{\VPzthEpq}}^{\UPMtlSUt 2}}
\newcommand {\hqLKZmNM} {\ensuremath{\Psi}}
\newcommand {\mpAYNUjM} [1] {\mcpkmsYW{\hqLKZmNM}{#1}}
\newcommand {\NxqVzdIP} {\ensuremath{\nuoeaWUW{W}^{+}}}
\newcommand {\NticoQpX} {\ensuremath{\nuoeaWUW{W}^{-}}}
\newcommand {\KlqDGJMb} {\ensuremath{\Pi^{+}}}
\newcommand {\ybAbpHUd} {\ensuremath{\Pi^{-}}}
\newcommand {\uvRBmLYf} {\ensuremath{k\evJAjzTZ{-SAT}}}
\iivqMZNr{\texorpdfstring
{One-Sided Error \ihRQQChH{} with Shared\\ EPR Pairs---A Simpler Proof}
{One-Sided Error QMA with Shared EPR Pairs---A Simpler Proof}}
\date{\formatdate{23}{6}{2013}}
\begin{document}
\maketitle
\begin{abstract}
We give a simpler proof of one of the results of
Kobayashi, Le Gall, and Nishimura \cite{Kobayashi2013},
which shows that any \ihRQQChH{} protocol can be
converted to a one-sided error protocol, in which Arthur
and Merlin initially share a constant number of EPR pairs
and then Merlin sends his proof to Arthur.
Our protocol is similar but somewhat simpler than the original.
Our main contribution is a simpler and more direct analysis
of the soundness property that uses well-known results
in quantum information
such as properties of the trace distance and the fidelity,
and the quantum \ExDgacVq{} theorem.
\end{abstract}
\section{Introduction}
The class \KRCqOJJJ{} was defined by Babai \cite{Babai1985}
as the natural probabilistic extension of the class \OYsGslWI.
In the definition of \KRCqOJJJ, the prover (Merlin) gives a polynomial
length `proof' to the verifier (Arthur), who then performs a
polynomial-time randomized computation and has to decide
if an input $x$ is in a language $L$ or not.
If we add interaction to the model, \YRlgvUEY{} the prover and
the verifier can exchange a polynomial number of messages
before the verifier makes his decision, then
we get the class \NEDasFOP{} \cite{Goldwasser1989}.\footnote{Babai also defined
an interactive version of \KRCqOJJJ, that can be thought of as a
`public-coin' version of \NEDasFOP.
Later Goldwasser and Sipser \cite{Goldwasser1986} showed
that this class has the same expressive power as \NEDasFOP.}
The verifiers of the above proof systems are allowed
to make some small error in their decision,
but they must satisfy two conditions.
\begin{itemize}
\item If $ x \in L $ then the verifier has to accept
a valid proof with high probability.
The probability that the verifier rejects such proof
is called the \emph{completeness} error.
\item If $ x \notin L $ then no matter what proof the verifier
receives, he must reject with high probability.
The probability that the verifier accepts an invalid
proof is called the \emph{soundness} error.
\end{itemize}
One of the first questions one may ask is whether it is
possible to get rid of one or both types of error.
It is easy to see that forcing the soundness error
to zero collapses \NEDasFOP{} (and also \KRCqOJJJ) to \OYsGslWI{} \cite{Arora2009}.
So we can't eliminate the soundness error completely,
but it is known that we can make it to be at most an
inverse-exponential function of the input length, without
reducing the expressive power of \KRCqOJJJ{} or \NEDasFOP.
On the other hand, it was shown by Zachos and F\"{u}rer \cite{Zachos1987}
that having \emph{perfect completeness}, also called
as \emph{one-sided error}, doesn't change the power of \KRCqOJJJ{}.
More formally, it holds that $ \KRCqOJJJ = \HnNIoaAl $,
where \HnNIoaAl{} is the class with perfect completeness.
The class \NEDasFOP{} can also be made to have one-sided error,
which follows, for example, from the characterization of \NEDasFOP{}
being equal to \fGxhDjlY, the class of problems decidable
in polynomial space
\cite{Lund1992,Shamir1992,Shen1992}.
For more information on these classes see \djIkZXTK{} the book
of Arora and Barak \cite{Arora2009}.
\par
Quantum Merlin-Arthur proof systems (and the class \ihRQQChH)
were introduced by Knill \cite{Knill1996}, Kitaev
\cite{Kitaev2002}, and also by Watrous \cite{Watrous2000}
as a natural extension
of \KRCqOJJJ{} and \OYsGslWI{} to the quantum computational setting.
Similarly, \wRzSJcrL{}s (and the class \GUYckNnU) were introduced by
Watrous \cite{Watrous2003} as a quantum
analogue of \NEDasFOP.
These classes have also been well studied and now it's known
that the power of \wRzSJcrL{}s is the same as the classical ones,
\YRlgvUEY{} $ \GUYckNnU = \NEDasFOP = \fGxhDjlY $ \cite{Jain2010}.
Furthermore, \wRzSJcrL{}s still have the same expressive power
if we restrict the number of messages to three and
have exponentially small one-sided error \cite{Kitaev2000}.
\par
The class \ihRQQChH{} is not as well understood as \GUYckNnU, but we
do have a reasonable amount of knowledge about it.
We know from the early results that it can be made
to have exponentially small two-sided error
\cite{Kitaev2002,Aharonov2002,Marriott2005}.
It also has natural complete problems, such as the
`$k$-local Hamiltonian' problem \cite{Kitaev2002,Aharonov2002},
for $ k \geq 2 $ \cite{Kempe2006},
which can be thought of as a quantum analogue of \uvRBmLYf.
With respect to the relation of \ihRQQChH{} to classical
complexity classes, we know that
$ \KRCqOJJJ \subseteq \ihRQQChH \subseteq \FEZMUXFw $
\cite{Marriott2005}.\footnote{A slightly stronger bound
of $ \ihRQQChH \subseteq \wkGMBCsG $ was shown by Vyalyi \cite{Vyalyi2003}.}
There are also interesting generalizations of \ihRQQChH,
such as with multiple unentangled provers
\cite{Kobayashi2003,Aaronson2009,Harrow2010,Blier2009},
but we will not consider them in this paper.
\par
Interestingly, we don't know if $ \ihRQQChH \HFkfLWeq \medaFzJc $, \YRlgvUEY{}
whether \ihRQQChH{} can be made to have perfect completeness.
It is a long-standing open problem which was already mentioned
in an early survey by Aharonov and Naveh \cite{Aharonov2002}.
Besides its inherent importance, giving a positive
answer to it would immediately imply that the
\medaFzJc-complete problems are also complete for \ihRQQChH.
Most notable of these is the `Quantum \uvRBmLYf' problem
of Bravyi \cite{Bravyi2006}, for $ k \geq 3 $ \cite{Gosset2013},
which is considered as a more natural quantum generalization
of \uvRBmLYf{} than the $k$-local Hamiltonian problem.\footnote{For
a list of \ihRQQChH- and \medaFzJc-complete problems,
see \djIkZXTK{} \cite{Bookatz2012}.}
Unfortunately, all previous techniques used to show one-sided
error properties of quantum interactive proof systems
require adding extra messages to the protocol
\cite{Kitaev2000,Kempe2008,Kobayashi2013},
so they can't be used directly in \ihRQQChH.
Aaronson \cite{Aaronson2009a} gave an evidence that shows
that proving $ \ihRQQChH = \medaFzJc $ may be difficult.
He proved that there exists a quantum oracle relative to which
$ \ihRQQChH \neq \medaFzJc $.
Another difficulty with \ihRQQChH, compared to \KRCqOJJJ,
is that in a \ihRQQChH{} proof system the acceptance
probability can be an arbitrary irrational number.
However, if certain assumptions are made
about the maximum acceptance probability then \ihRQQChH{}
can be made to have one-sided error \cite{Nagaj2009}.
Recently, Jordan, Kobayashi, Nagaj, and Nishimura
\cite{Jordan2012} showed that if Merlin's proof is
classical (in which case the class is denoted by \vPhVIqeI),
then perfect completeness is achievable, \YRlgvUEY{}
it holds that $ \vPhVIqeI = \HvMKrrDT $.
In another variant of \ihRQQChH, where we have multiple
unentangled provers and exponentially or double-exponentially
small gap, we also know that
perfect completeness is achievable \cite{Pereszlenyi2012}.
The most recent and strongest result towards proving
the original \ihRQQChH{} versus \medaFzJc{} question is by
Kobayashi, Le Gall, and Nishimura \cite{Kobayashi2013}.
They showed that we can convert a \ihRQQChH{} proof system to have
one-sided error, if we allow the prover and the verifier
of the resulting \medaFzJc{} protocol to share a
constant number of EPR pairs before the prover sends the
proof to the verifier.
The corresponding class is denoted by \hxlhxXVG.
With this notation, their result can be formalized
as the following theorem.
\begin{theorem}[\cite{Kobayashi2013}]
\label{GtTrXBeA}
$ \displaystyle \ihRQQChH \subseteq \hxlhxXVG $.
\end{theorem}
Since sharing an EPR pair can be done by
the verifier preparing it and sending half of it
to the prover, the above result implies that \ihRQQChH{}
is contained in the class of languages provable by
one-sided error, two-message quantum interactive proof systems
($ \ihRQQChH \subseteq \mcpkmsYW{\qOrOiFzp}{2} $).
This is a nontrivial upper bound.
Moreover, a result of Beigi, Shor, and Watrous
\cite{Beigi2011} implies that equality in
\XsBWhkhM{GtTrXBeA} holds,
resulting in the following characterization of \ihRQQChH.
\begin{corollary}[\cite{Kobayashi2013}]
$ \displaystyle \ihRQQChH = \hxlhxXVG = \bYQpVZwG $.
\end{corollary}
The \emph{contribution of this paper} is a conceptually simpler
and more direct proof of \XsBWhkhM{GtTrXBeA},
compared to the original one by Kobayashi \hFGJtSId\ \cite{Kobayashi2013}.
The algorithm of our verifier is also simpler, but the
main difference is in the proof of its soundness.
We believe that our proof helps to understand the result
better and we think that it may be simplified further.
The description of the idea behind our proof
can be found in \BeRtWGMB{RAEGhdYj},
while the complete proof is presented in
\BeRtWGMB{TJORLhCK}.
\subsection*{Organization of the Paper}
The remainder of the paper is organized as follows.
\BeRtWGMB{ZmECVrtM} discusses the background
definitions, theorems, and lemmas needed to understand
our proof.
The proof itself is presented in \BeRtWGMB{xpIEIUYC},
starting with a high level description in
\BeRtWGMB{RAEGhdYj}, and then presenting
the detailed proof in \BeRtWGMB{TJORLhCK}.
\section{Preliminaries}
\label{ZmECVrtM}
We assume familiarity with quantum information \cite{Watrous2008} and
computation \cite{Nielsen2000}; such as
quantum states, unitary operators, measurements,
quantum super-operators, etc.
We also assume the reader is familiar with
computational complexity, both classical \cite{Arora2009}
and quantum \cite{Watrous2008a}.
The purpose of this section is to present the notations
and background information (definitions,
theorems) required to understand the rest of the paper.
In this paper we denote the imaginary unit by \XDGqcIBa{} instead of
$i$, which we use as an index in summations, for example.
When we talk about a quantum register \FyfkKUDd{R}
of size $k$, we mean the object made up of
$k$ qubits.
It has associated \TbXXHkNG{}
$ \nuoeaWUW{R} = \ynpHhbyi^{2^k} $.
We always assume that some standard basis of
$ \nuoeaWUW{R} = \ynpHhbyi^{2^k} $ have been fixed and
we index those basis vectors by bit strings of length $k$.
So the standard basis of \nuoeaWUW{R} is denoted by
\LtfmKMob{\CARSeAIr{s}}{s \in \dVhhxuvI{0,1}^k}.
We denote the all zero string by
$ \LBGptwLj \aJAxrBMB 00 \ldots 0 $.
Throughout the paper,
\poaESbLo{\nuoeaWUW{R}} denotes the space of all linear
mappings from \nuoeaWUW{R} to itself.
The set of all density operators on \nuoeaWUW{R}
is denoted by \sYllRVEw{\nuoeaWUW{R}}.
The adjoint of $ \BWHXbaAC{A} \in \poaESbLo{\nuoeaWUW{R}} $
is denoted by \XhMHOPJK{\BWHXbaAC{A}}.
\begin{definition}
The \emph{trace norm} of $ \BWHXbaAC{A} \in \poaESbLo{\nuoeaWUW{R}} $
is defined by
\[ \JjPvCTWX{\BWHXbaAC{A}} \aJAxrBMB
\gyBwVFHt{\sqrt{\XhMHOPJK{\BWHXbaAC{A}} \BWHXbaAC{A}}}
\text{,} \]
and the \emph{operator norm} of \BWHXbaAC{A} is
\[ \bjhdjWcQ{\BWHXbaAC{A}} \aJAxrBMB
\BuwGPLRs{\max}{\eNNesbxu{\BWHXbaAC{A} \CARSeAIr{\varphi}}}
{\CARSeAIr{\varphi} \in \nuoeaWUW{R}, \VslhEMHJ
\eNNesbxu{\varphi} = 1} \text{.} \]
\end{definition}
The following inequality is a special case of the
H\"{o}lder Inequality for Schatten norms.
\begin{lemma}
\label{qtWTtRce}
For any \TbXXHkNG{} \nuoeaWUW{H} and operators
$ \BWHXbaAC{A}, \BWHXbaAC{B} \in \poaESbLo{\nuoeaWUW{H}} $,
it holds that
\[ \UylNyihM{\gyBwVFHt{\XhMHOPJK{\BWHXbaAC{B}} \BWHXbaAC{A}}}
\leq \JjPvCTWX{\BWHXbaAC{A}} \cdot \bjhdjWcQ{\BWHXbaAC{B}} . \]
\end{lemma}
The following definition is used to quantify the
distance between operators.
\begin{definition}
The \emph{trace distance} between operators
$ \BWHXbaAC{A}, \BWHXbaAC{B} \in \poaESbLo{\nuoeaWUW{H}} $
is defined as
\[ \Vohxuxyh{\BWHXbaAC{A}}{\BWHXbaAC{B}}
\aJAxrBMB \frac{\JjPvCTWX{\BWHXbaAC{A} - \BWHXbaAC{B}}}{2} . \]
If the operators represent pure quantum states,
\YRlgvUEY{} $ \BWHXbaAC{A} = \yAIHQVQC{\varphi} $
and $ \BWHXbaAC{B} = \yAIHQVQC{\psi} $, for some
$ \CARSeAIr{\varphi}, \CARSeAIr{\psi} \in \nuoeaWUW{H} $,
for which $ \eNNesbxu{\varphi} = \eNNesbxu{\psi} = 1 $,
then the trace distance
can be more conveniently written as
\begin{align}
\Vohxuxyh{\CARSeAIr{\varphi}}{\CARSeAIr{\psi}} =
\sqrt{1 - \UylNyihM{\mlpiBNcl{\varphi}{\psi}}^2} .
\label{uzRxUJas}
\end{align}
\end{definition}
Another way of quantifying the similarity between
density operators is by the fidelity defined below.
\begin{definition}
The \emph{fidelity} between
$ \rho, \sigma \in \sYllRVEw{\nuoeaWUW{H}} $
is defied as
\[ \DGaDmHJe{\rho}{\sigma} \aJAxrBMB
\JjPvCTWX{\sqrt{\rho} \sqrt{\sigma}} . \]
If $ \rho = \yAIHQVQC{\varphi} $ then the fidelity
can be more conveniently written as
\begin{align}
\DGaDmHJe{\yAIHQVQC{\varphi}}{\sigma}
= \sqrt{\bymhjSLm{\varphi} \sigma \CARSeAIr{\varphi}} .
\label{ZIjFQVNH}
\end{align}
\end{definition}
The following alternate characterization of the
fidelity will be useful later.
\begin{theorem}[Uhlmann's Theorem, see \djIkZXTK{} \cite{Watrous2008} for a proof]
\label{pQJzAmLP}
Let $ \rho, \sigma \in \sYllRVEw{\nuoeaWUW{H}} $ and
\nuoeaWUW{X} be a Hilbert space \GmcWftaL{}
$ \mcpkmsYW{\dim}{\nuoeaWUW{X}} \geq \mcpkmsYW{\dim}{\nuoeaWUW{H}} $.
Let $ \CARSeAIr{\varphi} \in \nuoeaWUW{X} \UPMtlSUt \nuoeaWUW{H} $
be any purification of $\rho$, \YRlgvUEY{}
$ \FbNPrGDa{\nuoeaWUW{X}}{\yAIHQVQC{\varphi}} = \rho $.
Then
\[ \DGaDmHJe{\rho}{\sigma} = \BuwGPLRs{\max}
{\UylNyihM{\mlpiBNcl{\varphi}{\psi}}}
{\CARSeAIr{\psi} \in \nuoeaWUW{X} \UPMtlSUt \nuoeaWUW{H}, \VslhEMHJ
\FbNPrGDa{\nuoeaWUW{X}}{\yAIHQVQC{\psi}} = \sigma} . \]
\end{theorem}
We now list some properties of the trace distance.
\begin{lemma}[triangle inequality]
\label{BhoQIhMu}
For any $ \BWHXbaAC{A}, \BWHXbaAC{B}, \BWHXbaAC{C}
\in \poaESbLo{\nuoeaWUW{H}} $,
it holds that
\[ \Vohxuxyh{\BWHXbaAC{A}}{\BWHXbaAC{B}}
\leq \Vohxuxyh{\BWHXbaAC{A}}{\BWHXbaAC{C}}
+ \Vohxuxyh{\BWHXbaAC{C}}{\BWHXbaAC{B}} . \]
\end{lemma}
\begin{theorem}[Theorem~9.2 from \cite{Nielsen2000}]
\label{NQlmRARP}
Let \rRGRugxs{\Phi}{\poaESbLo{\nuoeaWUW{H}}}{\poaESbLo{\nuoeaWUW{K}}}
be a quantum super-operator (a completely
positive and trace preserving linear map)
and let $ \rho, \sigma \in \sYllRVEw{\nuoeaWUW{H}} $.
Then
\[ \Vohxuxyh{\mcpkmsYW{\Phi}{\rho}}{\mcpkmsYW{\Phi}{\sigma}}
\leq \Vohxuxyh{\rho}{\sigma} . \]
\end{theorem}
\begin{lemma}
\label{XViefLtt}
Let $ \BWHXbaAC{A}, \BWHXbaAC{B} \in \poaESbLo{\nuoeaWUW{H}} $.
If $ 0 \leq \BWHXbaAC{B} $ and
$ \gyBwVFHt{\BWHXbaAC{B}} \leq \varepsilon $,
for some $ 0 \leq \varepsilon $, then
\[ \Vohxuxyh{\BWHXbaAC{A} + \BWHXbaAC{B}}{\BWHXbaAC{A}}
\leq \frac{\varepsilon}{2} . \]
\end{lemma}
\begin{proof}
From the definition of the trace norm and the trace
distance, together with the fact that
$ \sqrt{\XhMHOPJK{\BWHXbaAC{B}} \BWHXbaAC{B}} = \BWHXbaAC{B} $,
we get that
\begin{align*}
\Vohxuxyh{\BWHXbaAC{A} + \BWHXbaAC{B}}{\BWHXbaAC{A}}
&= \frac{\JjPvCTWX{\BWHXbaAC{A} + \BWHXbaAC{B} - \BWHXbaAC{A}}}{2} \\
&= \frac{\JjPvCTWX{\BWHXbaAC{B}}}{2} \\
&= \frac{\gyBwVFHt{\BWHXbaAC{B}}}{2} \\
&\leq \frac{\varepsilon}{2} . \GzyBBfAC
\end{align*}
\end{proof}
\begin{lemma}
\label{eJTZpBKB}
Let $ \rho, \sigma \in \sYllRVEw{\nuoeaWUW{H}} $
and $ 0 \leq \varepsilon < 1 $.
It holds that
\[ \Vohxuxyh{\GTIJeHWG{1 - \varepsilon} \rho
+ \varepsilon \sigma}{\rho} \leq \varepsilon . \]
\end{lemma}
\begin{proof}
Using the triangle inequality (\cEgZOrfG{BhoQIhMu})
and \cEgZOrfG{XViefLtt}, we get that
\begin{align*}
\Vohxuxyh{\GTIJeHWG{1 - \varepsilon} \rho
+ \varepsilon \sigma}{\rho}
&\leq \Vohxuxyh{\GTIJeHWG{1 - \varepsilon} \rho
+ \varepsilon \sigma}{\GTIJeHWG{1 - \varepsilon} \rho}
+ \Vohxuxyh{\rho}{\GTIJeHWG{1 - \varepsilon} \rho} \\
&\leq \frac{\varepsilon}{2} +
\frac{\JjPvCTWX{\rho - \GTIJeHWG{1 - \varepsilon} \rho}}{2} \\
&= \frac{\varepsilon}{2} +
\frac{\gyBwVFHt{\varepsilon \rho}}{2} \\
&= \varepsilon . \GzyBBfAC
\end{align*}
\end{proof}
The following lemma will be used to quantify
how much a projective measurement changes a state.
It is a variant of Winter's gentle measurement
lemma \cite{Winter1999}.
\begin{lemma}[Lemma~4 from \cite{Jain2012a}]
\label{kvwuafiV}
Let $ \rho \in \sYllRVEw{\nuoeaWUW{H}} $ be a density operator
and $ \Pi \in \poaESbLo{\nuoeaWUW{H}} $ be a projector
\GmcWftaL{} $ \gyBwVFHt{\rho \Pi} < 1 $.
Then
\[ 1 - \gyBwVFHt{\rho \Pi} \leq \DGaDmHJe{\rho}
{\frac{\GTIJeHWG{\aOIUkcRb - \Pi} \rho \GTIJeHWG{\aOIUkcRb - \Pi}}
{\gyBwVFHt{\rho \GTIJeHWG{\aOIUkcRb - \Pi}}}}^2 . \]
\end{lemma}
The following theorem gives a relation
between trace distance and fidelity.
\begin{theorem}[Fuchs-van de Graaf Inequalities, see \djIkZXTK{} \cite{Watrous2008} for a proof]
\label{gYublpMS}
For any $ \rho, \sigma \in \sYllRVEw{\nuoeaWUW{H}} $,
it holds that
\[ 1 - \Vohxuxyh{\rho}{\sigma}
\leq \DGaDmHJe{\rho}{\sigma} \leq
\sqrt{1 - \Vohxuxyh{\rho}{\sigma}^2} . \]
\end{theorem}
The following argument has appeared before,
for example in \cite{Beigi2011}.
We present it here as a separate lemma and
include its proof for convenience.
\begin{lemma}
\label{CpyQcgif}
Let $ 0 \leq \varepsilon \leq 1 $,
$ \rho \in \sYllRVEw{\nuoeaWUW{A} \UPMtlSUt \nuoeaWUW{B}} $,
and $ \sigma \in \sYllRVEw{\nuoeaWUW{B}} $.
If \[ \Vohxuxyh{\FbNPrGDa{\nuoeaWUW{A}}{\rho}}
{\sigma} \leq \varepsilon \]
then there exists a
$ \tau \in \sYllRVEw{\nuoeaWUW{A} \UPMtlSUt \nuoeaWUW{B}} $
for which
\begin{align*}
\FbNPrGDa{\nuoeaWUW{A}}{\tau} = \sigma
\qquad \text{and} \qquad
\Vohxuxyh{\rho}{\tau} \leq \sqrt{2 \varepsilon} .
\end{align*}
\end{lemma}
\begin{proof}
Let us take an auxiliary Hilbert space
$ \nuoeaWUW{X} \cong \nuoeaWUW{A} \UPMtlSUt \nuoeaWUW{B} $
and let $ \CARSeAIr{\varphi} \in \nuoeaWUW{X}
\UPMtlSUt \nuoeaWUW{A} \UPMtlSUt \nuoeaWUW{B} $ be a
purification of $\rho$, \YRlgvUEY{}
$ \FbNPrGDa{\nuoeaWUW{X}}{\yAIHQVQC{\varphi}} = \rho $.
We have that
\begin{align}
1 - \varepsilon &\leq 1 - \Vohxuxyh{
\FbNPrGDa{\nuoeaWUW{A}}{\rho}}{\sigma} \nonumber \\
&\leq \DGaDmHJe{\FbNPrGDa{\nuoeaWUW{A}}{\rho}}{\sigma}
\label{SuQkAthp} \\
&= \BuwGPLRs{\max}{\UylNyihM{\mlpiBNcl{\varphi}{\psi}}}
{\CARSeAIr{\psi} \in \nuoeaWUW{X} \UPMtlSUt \nuoeaWUW{A}
\UPMtlSUt \nuoeaWUW{B}, \VslhEMHJ
\FbNPrGDa{\nuoeaWUW{X} \UPMtlSUt \nuoeaWUW{A}}{\yAIHQVQC{\psi}}
= \sigma} \text{,}
\label{yeMdeEvS}
\end{align}
where \eqref{SuQkAthp} follows from
\XsBWhkhM{gYublpMS} and
\eqref{yeMdeEvS} follows
from \XsBWhkhM{pQJzAmLP}.
This means that there exists a
$ \CARSeAIr{\psi} \in \nuoeaWUW{X} \UPMtlSUt \nuoeaWUW{A} \UPMtlSUt \nuoeaWUW{B} $,
\GmcWftaL{} $ 1 - \varepsilon \leq
\UylNyihM{\mlpiBNcl{\varphi}{\psi}} $ and
$ \FbNPrGDa{\nuoeaWUW{X} \UPMtlSUt \nuoeaWUW{A}}{\yAIHQVQC{\psi}} = \sigma $.
Let
\[ \tau \aJAxrBMB \FbNPrGDa{\nuoeaWUW{X}}{\yAIHQVQC{\psi}} . \]
We only need to bound the distance between
$\rho$ and $\tau$.
\begin{align}
\Vohxuxyh{\rho}{\tau}
&\leq \Vohxuxyh{\CARSeAIr{\varphi}}{\CARSeAIr{\psi}}
\label{iqHFwUMb} \\
&= \sqrt{1 - \UylNyihM{\mlpiBNcl{\varphi}{\psi}}^2}
\label{KBmZvyLH} \\
&\leq \sqrt{1 - \GTIJeHWG{1 - \varepsilon}^2} \nonumber \\
&\leq \sqrt{2 \varepsilon} \text{,} \nonumber
\end{align}
where \eqref{iqHFwUMb} follows from
\XsBWhkhM{NQlmRARP}
and \eqref{KBmZvyLH}
follows from \eqref{uzRxUJas}.
\end{proof}
Throughout the paper we denote the identity operator
on some \TbXXHkNG{} \nuoeaWUW{H} by $ \aOIUkcRb_{\nuoeaWUW{H}} $ and
we sometimes omit the subscript if it is clear from the context.
We also use some well-known unitary operators
(also called quantum gates), such as the controlled-NOT
(\GECLKEJq) gate, the Hadamard gate (\eSvaqyjm), and the
Pauli operators (\FqTFSrwh, \RUYzfkFo, \WsoGCMAU).
The definition of these operators can be found in any
standard quantum textbook, for example in \cite{Nielsen2000}.
A key to our main algorithm will be the following
operator which will be used to reduce the acceptance
probability of a \ihRQQChH{} verifier to $1/2$.
The details will be explained later, but it is
convenient to define the operator here.
Let $ q \in \YnJXqMTX{0,1} $, then
\begin{align*}
\IwFBcXWd{q} \aJAxrBMB
\begin{bmatrix}
\sqrt{1 - q} & - \XDGqcIBa \sqrt{q} \\
- \XDGqcIBa \sqrt{q} & \sqrt{1 - q}
\end{bmatrix} .
\end{align*}
Note that \IwFBcXWd{q} corresponds to a rotation about
the $\hat{x}$ axes in the \vHMHwcRb{} and it is very
similar to the corresponding operator in \cite{Kobayashi2013}.
\par
We will use the following quantum states often
so it is convenient to introduce notations for them.
Let
\begin{align*}
\CARSeAIr{\MEVIdGJb} \aJAxrBMB
\frac{\CARSeAIr{0} + \CARSeAIr{1}}{\sqrt{2}}, \qquad
\CARSeAIr{\VPzthEpq} \aJAxrBMB
\frac{\CARSeAIr{0} - \CARSeAIr{1}}{\sqrt{2}}, \qquad
\CARSeAIr{\MEVIdGJb}, \CARSeAIr{\VPzthEpq} \in \ynpHhbyi^2 .
\end{align*}
Note that \CARSeAIr{\MEVIdGJb} and \CARSeAIr{\VPzthEpq} can
be obtained by applying \eSvaqyjm{} on \CARSeAIr{0} and \CARSeAIr{1}.
We will also use the Bell basis.
\begin{definition}
The following states form a basis of
$\ynpHhbyi^4$ and are called the \emph{Bell basis}.
\begin{align*}
\CARSeAIr{\Phi^{+}} &\aJAxrBMB \FGGdfwUg,
&\CARSeAIr{\Phi^{-}} &\aJAxrBMB \frac{\CARSeAIr{00} - \CARSeAIr{11}}{\sqrt{2}}, \\
\CARSeAIr{\Psi^{+}} &\aJAxrBMB \frac{\CARSeAIr{01} + \CARSeAIr{10}}{\sqrt{2}},
&\CARSeAIr{\Psi^{-}} &\aJAxrBMB \frac{\CARSeAIr{01} - \CARSeAIr{10}}{\sqrt{2}}.
\end{align*}
\end{definition}
The following theorem is used to eliminate
the entanglement between registers.
\begin{theorem}[quantum \ExDgacVq{} theorem \cite{Christandl2007};
this form is from \cite{Watrous2008}\footnote{Note
that this is not the general form of the theorem,
but this simplified version will be sufficient
for our proof.}]
\label{cmMerJQG}
Let $ \FyfkKUDd{X}_1, \ldots, \FyfkKUDd{X}_{\gAuWdsUU} $
be identical quantum registers,
each having associated space $\ynpHhbyi^2$,
and let $ \rho \in \sYllRVEw{ \ynpHhbyi^{2 \gAuWdsUU} } $
be the state of these registers.
Suppose that $\rho$ is invariant under the permutation
of the registers.
Then there exist a number $ m \in \FLbfwYzZ $,
a probability distribution
\LtfmKMob{p_i}{i \in \hdzgZeCW{m}},
and a collection of density operators
$ \LtfmKMob{\xi_i}{i \in \hdzgZeCW{m}} \subset \sYllRVEw{\ynpHhbyi^2} $
such that
\[ \JjPvCTWX{\FbNPrGDa{\nuoeaWUW{X}_3, \ldots, \nuoeaWUW{X}_{\gAuWdsUU}}{\rho}
- \sum_{i=1}^{m} p_i \xi_i \UPMtlSUt \xi_i}
< \frac{32}{\gAuWdsUU} . \]
\end{theorem}
Later we will use the \KlrOGnZA{} of \cite{Barenco1997,Buhrman2001}
and the following property of this test.
\begin{theorem}[\cite{Buhrman2001,Kobayashi2003}]
\label{BNEmInIr}
When the \KlrOGnZA{} is applied to
$ \rho \UPMtlSUt \sigma $, where
$ \rho, \sigma \in \sYllRVEw{\nuoeaWUW{H}} $,
it succeeds with probability
\[ \frac{1 + \gyBwVFHt{\rho \sigma}}{2} . \]
\end{theorem}
In order to perform the \KlrOGnZA, we need two Hadamard gates,
\WrVYYpex{\mcpkmsYW{\log}{\mcpkmsYW{\dim}{\nuoeaWUW{H}}}}-number of
\GECLKEJq{} gates,
and we need to measure a qubit in the standard basis.
\par
The following lemma will be the basic building block
to prove perfect completeness, similarly to \cite{Kobayashi2013}.
\begin{lemma}
\label{bszdFjVr}
Let $ \Delta, \Pi \in \poaESbLo{\nuoeaWUW{H}} $
be projectors.
Suppose that one of the eigenvalues of
$ \Delta \Pi \Delta $ is $1/2$ with
corresponding eigenstate \CARSeAIr{\omega}.
Then
\[ \Delta \GTIJeHWG{\aOIUkcRb - 2 \Pi} \Delta \CARSeAIr{\omega} = 0 . \]
\end{lemma}
\begin{proof}
Using the fact that
$ \Delta \CARSeAIr{\omega} = \CARSeAIr{\omega} $,
we get
\begin{align*}
\Delta \GTIJeHWG{\aOIUkcRb - 2 \Pi} \Delta \CARSeAIr{\omega}
&= \GTIJeHWG{\Delta - 2 \Delta \Pi \Delta} \CARSeAIr{\omega} \\
&= \CARSeAIr{\omega} - 2 \GTIJeHWG{\frac{1}{2} \CARSeAIr{\omega}} \\
&= 0 . \GzyBBfAC
\end{align*}
\end{proof}
In \cite{Kobayashi2013}, the procedure defined by
applying $ \Delta \GTIJeHWG{\aOIUkcRb - 2 \Pi} \Delta $ is called
`\gorjOgEG'.
The procedure is very similar to the quantum rewinding
technique of Watrous \cite{Watrous2009}, which
has been used before to achieve perfect completeness
for quantum multi-prover interactive proofs \cite{Kempe2008}.
Also note that the idea behind the quantum rewinding
technique dates back to the strong gap amplification
for \ihRQQChH{} \cite{Marriott2005}.
\par
It should be mentioned here that \cEgZOrfG{bszdFjVr}
will only be used in the honest case, while in the
dishonest case we will argue about the rejection
probability directly.
This is why we can have a much simpler lemma
compared to the description of the \gorjOgEG{}
in \cite{Kobayashi2013}.
\subsection{\CTcThSvb{} Representations and Post-Selection}
\label{XoAYconU}
Let \rRGRugxs{\Phi}{\poaESbLo{\ynpHhbyi^2}}{\poaESbLo{\ynpHhbyi^2}}
be a quantum super-operator (a completely
positive and trace preserving linear map).
The normalized \dFUQloAO{} of $ \Phi $ is defined as\footnote{The
\dFUQloAO{} is obviously defined for any dimension, but
in this paper we will only need it for qubits, so
we will be fine with this restricted definition.}
\[ \rho_{\Phi} \aJAxrBMB
\frac{1}{2} \sum_{ x,y \in \dVhhxuvI{0,1} }
\mcpkmsYW{\Phi}{\SCAQAFmU{x}{y}} \UPMtlSUt \SCAQAFmU{x}{y} \text{,}
\qquad \rho_{\Phi} \in \sYllRVEw{\ynpHhbyi^4} \text{.} \]
Suppose we have an EPR pair (\CARSeAIr{\Phi^{+}}) in registers
\GTIJeHWG{\FyfkKUDd{S}, \FyfkKUDd{S}'}.
Then $\rho_{\Phi}$ can be generated by applying $\Phi$
on register \FyfkKUDd{S}.
If $\Phi$ is unitary, \YRlgvUEY{}
$ \mcpkmsYW{\Phi}{\sigma} = \XhMHOPJK{\BWHXbaAC{U}}
\sigma \BWHXbaAC{U} $, for some unitary operator
\BWHXbaAC{U}, then $\rho_{\Phi}$ is pure, in which
case we use the notation \qIsFHsNm{\BWHXbaAC{U}},
where $ \yAIHQVQC{\mcpkmsYW{J}{\BWHXbaAC{U}}} = \rho_{\Phi} $.
Let $ q \in \YnJXqMTX{0,1} $.
By simple calculation, we get that
\begin{align*}
\qIsFHsNm{\IwFBcXWd{q}} &=
\GTIJeHWG{\IwFBcXWd{q} \UPMtlSUt \aOIUkcRb} \CARSeAIr{\Phi^{+}} =
\sqrt{1 - q} \CARSeAIr{\Phi^{+}} - \XDGqcIBa \sqrt{q} \CARSeAIr{\Psi^{+}}
\text{,} \\
\qIsFHsNm{\XhMHOPJK{\IwFBcXWd{q}}} &=
\GTIJeHWG{\XhMHOPJK{\IwFBcXWd{q}} \UPMtlSUt \aOIUkcRb} \CARSeAIr{\Phi^{+}} =
\sqrt{1 - q} \CARSeAIr{\Phi^{+}} + \XDGqcIBa \sqrt{q} \CARSeAIr{\Psi^{+}} .
\end{align*}
\par
In \nvBBVayi{alg:QMA_one verifier}, on
page~\pageref{alg:QMA_one verifier}, we will be given two
copies of \qIsFHsNm{\XhMHOPJK{\IwFBcXWd{q}}} and we will
have to create the state $ \IwFBcXWd{q} \CARSeAIr{0} $ with
the help of the first copy.
Using the second copy, we will need to apply
\XhMHOPJK{\IwFBcXWd{q}} on an arbitrary input state.
The way these can be done is as follows.
Suppose now that we are given \qIsFHsNm{\XhMHOPJK{\IwFBcXWd{q}}}
and we want to create $ \IwFBcXWd{q} \CARSeAIr{0} $.
This can easily be done by applying the following unitary
\begin{align}
\cNXPyGJX \aJAxrBMB \SCAQAFmU{00}{\Phi^{+}}
- \SCAQAFmU{10}{\Psi^{+}} + \SCAQAFmU{01}{\Phi^{-}}
- \SCAQAFmU{11}{\Psi^{-}} \text{,}
\label{yDshCAyj}
\end{align}
because $ \cNXPyGJX \qIsFHsNm{\XhMHOPJK{\IwFBcXWd{q}}} =
\GTIJeHWG{\IwFBcXWd{q} \CARSeAIr{0}} \UPMtlSUt \CARSeAIr{0} $.
Now assume that we want to apply \XhMHOPJK{\IwFBcXWd{q}}
on an arbitrary state \CARSeAIr{\varphi}, with
the help of \qIsFHsNm{\XhMHOPJK{\IwFBcXWd{q}}}.
This can be accomplished with probability $1/2$
by a procedure that we call post-selection.
The procedure is described in \nvBBVayi{alg:post selection}.
Note that \nvBBVayi{alg:post selection} is basically
teleportation, where we want to teleport the state
of \FyfkKUDd{X} (let's say it's \CARSeAIr{\varphi})
to register \FyfkKUDd{S}.
If we get output \CARSeAIr{\Phi^{+}} then no correction is
needed in the teleportation.
Since \XhMHOPJK{\IwFBcXWd{q}} was applied to \FyfkKUDd{S}
before, we get $ \XhMHOPJK{\IwFBcXWd{q}} \CARSeAIr{\varphi} $
in \FyfkKUDd{S}.
If the output is \CARSeAIr{\Psi^{+}} then there is a
`Pauli-\FqTFSrwh{} error' in the teleportation
so we get $ \XhMHOPJK{\IwFBcXWd{q}} \FqTFSrwh \CARSeAIr{\varphi} $,
which we can correct since \XhMHOPJK{\IwFBcXWd{q}} and
\FqTFSrwh{} commute.
In case of the other two outputs (\CARSeAIr{\Phi^{-}}
and \CARSeAIr{\Psi^{-}}), there is a \RUYzfkFo{} or
a \WsoGCMAU{} error that we can't correct, so we
declare failure.
This idea of simulating a quantum operator with
\dFUQloAO{}s has appeared before in the context of
quantum interactive proof and quantum
Merlin-Arthur proof systems, such as in
Refs.~\cite{Beigi2011,Kobayashi2013}.
We state a lemma here that we will use in the
honest case.
In the dishonest case, we will
argue about the success probability and the
output of \nvBBVayi{alg:post selection} in the
analysis of \nvBBVayi{alg:QMA_one verifier}.
\UDODcOFW[tb]{Post-Selection}{post selection}{
\rNokgjXa single qubit registers \FyfkKUDd{S},
$\FyfkKUDd{S}'$, \FyfkKUDd{X}
\EfCTdHOd{\GTIJeHWG{\FyfkKUDd{S}, \FyfkKUDd{S}'} are supposed
to contain the state \qIsFHsNm{\XhMHOPJK{\IwFBcXWd{q}}}.}
\BeqmjZec success and \FyfkKUDd{S}, or failure
\STATE Perform a measurement in the Bell basis on
\GTIJeHWG{\FyfkKUDd{S}', \FyfkKUDd{X}}.
\IF{the output is \CARSeAIr{\Phi^{+}}}
\RETURN success and \FyfkKUDd{S}
\ELSIF{the output is \CARSeAIr{\Psi^{+}}}
\STATE Apply \FqTFSrwh{} on \FyfkKUDd{S}.
\RETURN success and \FyfkKUDd{S}
\ELSE
\RETURN failure
\ENDIF
}
\begin{lemma}
\label{dUOsaZJF}
Suppose that the inputs to \nvBBVayi{alg:post selection}
are \qIsFHsNm{\XhMHOPJK{\IwFBcXWd{q}}} in
\GTIJeHWG{\FyfkKUDd{S}, \FyfkKUDd{S}'}, for some
$ q \in \YnJXqMTX{0,1} $, and an arbitrary \CARSeAIr{\varphi}
in \FyfkKUDd{X}.
Then the algorithm will succeed with probability
$1/2$ and in that case it will output
$ \XhMHOPJK{\IwFBcXWd{q}} \CARSeAIr{\varphi} $ in \FyfkKUDd{S}.
\end{lemma}
\subsection{Quantum Merlin-Arthur Proof Systems}
\label{oTrRpwia}
Before we define the complexity class \ihRQQChH,
let us briefly describe what we mean by
polynomial-time quantum algorithms or quantum verifiers.
Quantum verifiers are polynomial-time uniformly
generated quantum circuits consisting of some
universal set of gates.
There are many different universal sets and we assume
that one of them has been chosen beforehand.
Usually it doesn't matter which set we choose when we
define quantum verifiers and classes like \DYmPyZgR{} or \ihRQQChH,
because it is known that each universal set can approximate
any other set with exponential precision.
However, in the paper we will have quantum proof
systems with one-sided error, in which case the
gate set may matter.
This is because simulating one set of gates with
another may ruin the one-sided error property.
In this paper, we only assume that the verifier
can perform or perfectly simulate the \GECLKEJq{}
and the \eSvaqyjm{} gate with his universal set,
besides being able to perform any
polynomial-time classical computation.
Note that with \GECLKEJq{} and \eSvaqyjm, one can perform
all Pauli operators, as well as operator \cNXPyGJX,
defined by Eq.~\eqref{yDshCAyj}.
The above assumption is enough for our result, so
we won't bother about the gate set in the rest of the paper.
\begin{definition}[\cite{Watrous2000,Aharonov2002}]
\label{IQpYyyBm}
For functions \rRGRugxs{c, s}{\FLbfwYzZ}{ \left( 0, 1 \right] },
a language $L$ is in \pbpWqYJX{c}{s} if there
exists a quantum verifier \jYpQwXZd{V} with
the following properties.
For all $ n \in \FLbfwYzZ $ and inputs $ x \in \dVhhxuvI{0,1}^n $,
the circuit of \jYpQwXZd{V} on input $x$,
denoted by \LZfegDul{}, is a polynomial-time uniformly
generated quantum circuit acting on two
polynomial-size registers \FyfkKUDd{P} and \FyfkKUDd{A}.
One output qubit of \LZfegDul{} is designated as the
acceptance qubit.
We say that \LZfegDul{} on input $ \CARSeAIr{\varphi}_{\nuoeaWUW{P}}
\UPMtlSUt \CARSeAIr{\LBGptwLj}_{\nuoeaWUW{A}} $ accepts if
the acceptance qubit of $ \LZfegDul \GTIJeHWG{\CARSeAIr{\varphi}_{\nuoeaWUW{P}}
\UPMtlSUt \CARSeAIr{\LBGptwLj}_{\nuoeaWUW{A}}} $ is
projected to \CARSeAIr{1}
and we say that \LZfegDul{} rejects if it's projected to \CARSeAIr{0}.
\LZfegDul{} must satisfy the following properties.
\begin{description}
\item[Completeness.]
If $x \in L$ then there exists a quantum state
$ \CARSeAIr{\varphi} \in \nuoeaWUW{P} $ \GmcWftaL{} the acceptance
probability of \LZfegDul, on input
$ \CARSeAIr{\varphi} \UPMtlSUt \CARSeAIr{\LBGptwLj}_{\nuoeaWUW{A}} $,
is at least \mcpkmsYW{c}{n}.
\item[Soundness.]
If $x \notin L$ then for all states
$ \CARSeAIr{\varphi} \in \nuoeaWUW{P} $, \LZfegDul{}
accepts with probability at most \mcpkmsYW{s}{n},
given $ \CARSeAIr{\varphi} \UPMtlSUt \CARSeAIr{\LBGptwLj}_{\nuoeaWUW{A}} $
as its input.
\end{description}
\end{definition}
Note that \FyfkKUDd{P} is the register
in which the verifier receives his proof
and \FyfkKUDd{A} is his private register,
which is, \LDYwFcAX, always initialized to
\CARSeAIr{\LBGptwLj}.
Without causing confusion, we will denote both
the circuit of the verifier and the unitary operator
it represents by \LZfegDul.
\begin{definition}
The class \ihRQQChH{} is defined as
$ \ihRQQChH \aJAxrBMB \pbpWqYJX{\frac{2}{3}}{\frac{1}{3}} $.
\end{definition}
The choice of the constants
in the above definition are arbitrary,
as shown by the following theorem.
\begin{theorem}[\cite{Kitaev2002,Aharonov2002,Marriott2005}]
\label{bmJxjwbd}
Let $ c \in \GTIJeHWG{0,1} $ be a constant and
\mcpkmsYW{p}{n} be a positive polynomial in $n$.
It holds that
\begin{align*}
\ihRQQChH &= \pbpWqYJX{c}{c - \frac{1}{\mcpkmsYW{p}{n}}}
= \pbpWqYJX{1 - 2^{- \mcpkmsYW{p}{n}}}{2^{- \mcpkmsYW{p}{n}}} .
\end{align*}
\end{theorem}
\begin{definition}
The class \mjIKlHwf{c}{s} is defined the same way as
\pbpWqYJX{c}{s} in \ozsaYFnP{IQpYyyBm},
except that before the prover sends the proof to the
verifier, they can share a constant number of EPR pairs
(the two-qubit state \CARSeAIr{\Phi^{+}}).
\end{definition}
\begin{definition}
The class \hxlhxXVG{} is defined as
$ \hxlhxXVG \aJAxrBMB \mjIKlHwf{1}{1/2} $.
\end{definition}
Similarly as before, the choice of $1/2$ is arbitrary.
This is because a \hxlhxXVG{} proof system is a special
case of a two-message \qOrOiFzp{} proof system
and perfect parallel repetition holds even for
three-message \qOrOiFzp{} \cite{Kitaev2000}.
So we have the following lemma.
\begin{lemma}
\label{PpcvCrEf}
Let \mcpkmsYW{p}{n} be a positive polynomial in $n$.
It holds that
\[ \hxlhxXVG = \mjIKlHwf{1}{1 - \frac{1}{\mcpkmsYW{p}{n}}}
= \mjIKlHwf{1}{2^{- \mcpkmsYW{p}{n}}} . \]
\end{lemma}
\section{Proof of \XsBWhkhM{GtTrXBeA}}
\label{xpIEIUYC}
Before we give the detailed proof of
\XsBWhkhM{GtTrXBeA}, let us
briefly describe the intuition behind our proof.
We also point out the similarities and the differences
between our proof and the proof in \cite{Kobayashi2013}.
\subsection{The Idea Behind the Proof}
\label{RAEGhdYj}
The basic idea to achieve perfect completeness
is very similar to Ref.~\cite{Kobayashi2013}.
For any input $x$, let us define
\[ \lmkGSgYf \aJAxrBMB \GTIJeHWG{\IbAozRcX} \XhMHOPJK{\LZfegDul}
\ulbwWGDE \LZfegDul \GTIJeHWG{\IbAozRcX} \text{,} \]
where \LZfegDul{} is the same as in \BeRtWGMB{oTrRpwia}
and \ulbwWGDE{} is the projector that corresponds to projecting
the acceptance qubit of \LZfegDul{} to \CARSeAIr{1}.
Note that $ 0 \leq \lmkGSgYf \leq \aOIUkcRb_{\nuoeaWUW{P} \UPMtlSUt \nuoeaWUW{A}} $.
As was observed in \cite{Marriott2005}, the maximum
acceptance probability of \LZfegDul{} is \bjhdjWcQ{\lmkGSgYf},
or in other words, the maximum eigenvalue of \lmkGSgYf.
We will use \cEgZOrfG{bszdFjVr}
to construct a test that succeeds with probability $1$
in case $ x \in L $.
In order to achieve this, we need that for all $ x \in L $,
$ \bjhdjWcQ{\lmkGSgYf} = 1/2 $.
Unfortunately, this is not true in general.
Instead, we have that if $ x \in L $ then $ \bjhdjWcQ{\lmkGSgYf} \geq 1/2 $.
Our first objective is to modify \lmkGSgYf{} \GmcWftaL{} its maximum
eigenvalue is exactly $1/2$.
We do this by using an auxiliary qubit
(stored in register \FyfkKUDd{S}) and defining
\begin{align*}
\lmkGSgYf' &\aJAxrBMB \lmkGSgYf \UPMtlSUt \GTIJeHWG{\yAIHQVQC{0}_{\nuoeaWUW{S}}
\XhMHOPJK{\IwFBcXWd{q}} \yAIHQVQC{1}_{\nuoeaWUW{S}} \IwFBcXWd{q}
\yAIHQVQC{0}_{\nuoeaWUW{S}}} \\
&= \GTIJeHWG{\aOIUkcRb_{\nuoeaWUW{P}} \UPMtlSUt
\yAIHQVQC{\LBGptwLj}_{\nuoeaWUW{A} \UPMtlSUt \nuoeaWUW{S}}}
\XhMHOPJK{\GTIJeHWG{\LZfegDul \UPMtlSUt \IwFBcXWd{q}}}
\GTIJeHWG{\ulbwWGDE \UPMtlSUt \yAIHQVQC{1}_{\nuoeaWUW{S}}}
\GTIJeHWG{\LZfegDul \UPMtlSUt \IwFBcXWd{q}}
\GTIJeHWG{\aOIUkcRb_{\nuoeaWUW{P}} \UPMtlSUt
\yAIHQVQC{\LBGptwLj}_{\nuoeaWUW{A} \UPMtlSUt \nuoeaWUW{S}}}
\text{,}
\end{align*}
where $ q \aJAxrBMB \frac{1}{2p} \in \YnJXqMTX{\frac{1}{2}, 1} $.
It is now easy to see that $ \bjhdjWcQ{\lmkGSgYf'} = 1/2 $
and we can also write $\lmkGSgYf'$ as $ \lmkGSgYf' = \Delta \Pi \Delta $, for
\begin{align*}
\Delta \aJAxrBMB \aOIUkcRb_{\nuoeaWUW{P}} \UPMtlSUt
\yAIHQVQC{\LBGptwLj}_{\nuoeaWUW{A} \UPMtlSUt \nuoeaWUW{S}}
\qquad \text{and} \qquad
\Pi \aJAxrBMB \XhMHOPJK{\GTIJeHWG{\LZfegDul \UPMtlSUt \IwFBcXWd{q}}}
\GTIJeHWG{\ulbwWGDE \UPMtlSUt \yAIHQVQC{1}_{\nuoeaWUW{S}}}
\GTIJeHWG{\LZfegDul \UPMtlSUt \IwFBcXWd{q}} .
\end{align*}
Now, we can use \cEgZOrfG{bszdFjVr}
and obtain the following test.
Let the principal eigenvector of $\lmkGSgYf'$ (that corresponds to eigenvalue
$1/2$) be denoted by $ \CARSeAIr{\omega}_{\nuoeaWUW{P}}
\UPMtlSUt \CARSeAIr{\LBGptwLj}_{\nuoeaWUW{A} \UPMtlSUt \nuoeaWUW{S}}$.
The test receives this eigenstate as the input,
applies the unitary operator $ \aOIUkcRb - 2 \Pi $,
and performs a measurement defined by operators \dVhhxuvI{\Delta, \aOIUkcRb - \Delta}.
If the state is projected to $\Delta$ the test rejects and
otherwise it accepts.
\cEgZOrfG{bszdFjVr} guarantees that
we never project to $\Delta$.
\par
However, a polynomial-time verifier may not be able to perform
this test, because it is possible that \IwFBcXWd{q}
can't be expressed by a polynomial-size quantum circuit and
the verifier may not even know the exact value of $q$.
To overcome this difficulty, the verifier expects
the prover to give several copies of the
normalized \dFUQloAO{}s of \XhMHOPJK{\IwFBcXWd{q}},
besides $ \CARSeAIr{\omega}_{\nuoeaWUW{P}} $.
As explained in \BeRtWGMB{XoAYconU},
these can be used to perform \IwFBcXWd{q} and \XhMHOPJK{\IwFBcXWd{q}},
by using unitary \cNXPyGJX{} to do \IwFBcXWd{q}, and
\nvBBVayi{alg:post selection} to do \XhMHOPJK{\IwFBcXWd{q}}.
Note that \nvBBVayi{alg:post selection} may fail, in which
case we have to accept in order to maintain perfect completeness.
This is the main idea to prove perfect completeness,
and it is basically the same as in \cite{Kobayashi2013}.
\par
The harder part is to prove the soundness and this is where
our proof differs from the one in \cite{Kobayashi2013}.
Let us first give a high-level overview of the soundness
proof of Kobayashi \hFGJtSId\ \cite{Kobayashi2013}.
The main idea in their proof is to perform a
sequence of tests (\YRlgvUEY{} quantum algorithms with
measurements at the end), which together ensure
that the registers that are supposed to contain the
\dFUQloAO{}s of the desired operator, actually
contain the \dFUQloAO{}s of \emph{some} operator.
Then they show that doing the so-called
`\LkIoSNVy', the one just described above,
with these states in the registers, will
cause rejection with some constant probability.
The tests they use to ensure that the states are
close to \dFUQloAO{}s are the `\hhqeGUPm' (which is used to
remove the entanglement between the register of the original
proof and the registers of the \dFUQloAO{}s), the `\mCDuGkWD'
(which tests that the states are in a certain subspace),
and the \KlrOGnZA.
In their analysis they also use the \ExDgacVq{} theorem.
We don't describe these tests here, as the interested
reader can find them in \cite{Kobayashi2013}.
We just list them in order to compare them to
the tools we use.
\par
Our main idea behind the soundness proof is conceptually
different.
We don't argue that the states are close to \dFUQloAO{}s, but
we analyze our version of the \LkIoSNVy{} directly.
As we described this test above, there are two
measurements in it.
The first measurement is in \nvBBVayi{alg:post selection}
and the second is given by \dVhhxuvI{\Delta, \aOIUkcRb - \Delta}.
So, roughly speaking, we have to prove two things.
First, we have to show that \nvBBVayi{alg:post selection}
can't always fail, as otherwise we would end up
always accepting without reaching the end of the procedure.
This will be formalized later in \cEgZOrfG{oBbNiOJY}.
In order to prove \cEgZOrfG{oBbNiOJY},
we only need two assumptions.
The first assumption is that the state being
measured in \nvBBVayi{alg:post selection}
is separable, which is guaranteed by the \ExDgacVq{} theorem
(\XsBWhkhM{cmMerJQG}).
The second assumption is that the state of some registers
is close to being completely mixed, which is obviously
true because these registers hold parts of EPR pairs.
\par
The second part of the soundness proof is to show
that conditioned on \nvBBVayi{alg:post selection} being successful,
we get a state that projects to $\Delta$ with constant
probability.
To prove this, we first argue that the private register
of the verifier (register \FyfkKUDd{A}) projects
to \yAIHQVQC{\LBGptwLj}.
This follows from simple properties of the trace distance.
We then show that the state of register \FyfkKUDd{S} projects to \yAIHQVQC{0}.
To prove this, we use the \KlrOGnZA{} on the registers
that are supposed to contain the \dFUQloAO{}s.
This ensures that the state of these registers are close
to the same pure state.
This property is formalized in \cEgZOrfG{ijiuOsQF}.
We also use a simplified version of the \mCDuGkWD,
which is not really a test but an application of
a super-operator on the above mentioned registers.
This super-operator will be defined later in
Eq.~\eqref{VpYNOSVe}.
We can think of it as performing a projective
measurement that corresponds to the \mCDuGkWD{}
and forgetting the outcome.
Using the above tools, it will follow by direct
calculation that the state of \FyfkKUDd{S} projects
to \yAIHQVQC{0}.
\par
Note that we don't use the \hhqeGUPm{} of \cite{Kobayashi2013}
and we use a simpler form of the \mCDuGkWD.
Besides that, it's worth mentioning that the tools we use
can be grouped into two sets based on whether we use them
in the analysis of the first or the second measurement.
For the analysis of the first measurement, we need that
some state is close to being maximally mixed, while in the
analysis of the second, we use the \KlrOGnZA{} and
the above mentioned super-operator.
One exception is the \ExDgacVq{} theorem, as we need that
the states are separable in both parts.
This property of the proof may be useful for simplifying
it further, because for example, to omit the \KlrOGnZA, one
would only need to re-prove that the state of \FyfkKUDd{S} projects
to \yAIHQVQC{0} in the last measurement.
\subsection{The Detailed Proof}
\label{TJORLhCK}
This section presents the detailed proof of
\XsBWhkhM{GtTrXBeA}.
Let $ L \in \ihRQQChH $ and \jYpQwXZd{V} be the corresponding
verifier.
Let $x$ be an input to language $L$ and let us denote
its length by $n$.
We denote the circuit of \jYpQwXZd{V} on input $x$
(and also the unitary transformation it represents)
by \LZfegDul.
Let the private register of \LZfegDul{} be denoted by
\FyfkKUDd{A} and the register in which the proof is
received by \FyfkKUDd{P}.
As in the previous section, let $ \ulbwWGDE \in
\poaESbLo{\nuoeaWUW{P} \UPMtlSUt \nuoeaWUW{A}} $ be the
projector that corresponds to projecting
the acceptance qubit of \LZfegDul{} to \CARSeAIr{1}.
By \XsBWhkhM{bmJxjwbd}, we assume
that the completeness of \jYpQwXZd{V} is
at least $ 1/2 $ and his soundness is at most
$ 4^{-n} $.
Let $ \gAuWdsUU \aJAxrBMB 2^{107} $.
We construct a verifier \jYpQwXZd{W} which recognizes
the same language $L$ with completeness $1$,
constant soundness, and with the additional property
that \jYpQwXZd{W} possesses \gAuWdsUU{} halves of
EPR pairs in registers $\FyfkKUDd{S}_1',
\ldots, \FyfkKUDd{S}_{\gAuWdsUU}'$ before the protocol begins.
The other halves of the EPR pairs are held by the prover.
\jYpQwXZd{W} gets his proof in registers
$\FyfkKUDd{P}, \FyfkKUDd{S}_1, \ldots, \FyfkKUDd{S}_{\gAuWdsUU}$,
where the $ \FyfkKUDd{S}_i $'s are single qubit registers,
which had contained the other halves of the EPR pairs
before the prover performed some transformation on them.
\jYpQwXZd{W} expects to get the original proof of \jYpQwXZd{V}
in \FyfkKUDd{P} and the state of each
\GTIJeHWG{\FyfkKUDd{S}_i, \FyfkKUDd{S}_i'} is supposed to be
\qIsFHsNm{\XhMHOPJK{\IwFBcXWd{q}}}, for some $ q \in \YnJXqMTX{0,1} $.
In the description of \jYpQwXZd{W} we will use
the following notations.
Let \NxqVzdIP{} be the subspace of $\ynpHhbyi^4$ spanned by
\CARSeAIr{\Phi^{+}} and \CARSeAIr{\Psi^{+}}, and \NticoQpX{}
be the subspace spanned by \CARSeAIr{\Phi^{-}} and \CARSeAIr{\Psi^{-}}.
Let
\begin{align*}
\KlqDGJMb \aJAxrBMB \yAIHQVQC{\Phi^{+}} + \yAIHQVQC{\Psi^{+}}
\qquad \text{and} \qquad
\ybAbpHUd \aJAxrBMB \yAIHQVQC{\Phi^{-}} + \yAIHQVQC{\Psi^{-}}
\text{,}
\end{align*}
\YRlgvUEY{} the projections to subspaces \NxqVzdIP{} and
\NticoQpX.
Let \rRGRugxs{\hqLKZmNM}{\poaESbLo{\ynpHhbyi^4}}{\poaESbLo{\ynpHhbyi^4}}
be a quantum super-operator defined as
\begin{align}
\mpAYNUjM{\BWHXbaAC{A}} \aJAxrBMB
\KlqDGJMb \BWHXbaAC{A} \KlqDGJMb +
\ybAbpHUd \BWHXbaAC{A} \ybAbpHUd .
\label{VpYNOSVe}
\end{align}
\cNXPyGJX{} still denotes the operator
defined by Eq.~\eqref{yDshCAyj}.
With these notations, the procedure of \jYpQwXZd{W} is described in
\nvBBVayi{alg:QMA_one verifier}.
\par
\UDODcOFW[tb]{Description of verifier \jYpQwXZd{W} in the proof of
\XsBWhkhM{GtTrXBeA}.}
{QMA_one verifier}
{
\rNokgjXa description of a circuit \LZfegDul,
polynomial-size register \FyfkKUDd{P}
compatible with \LZfegDul, and
single qubit registers $ \FyfkKUDd{S}_1,
\ldots, \FyfkKUDd{S}_{\gAuWdsUU}, \FyfkKUDd{S}_1',
\ldots, \FyfkKUDd{S}_{\gAuWdsUU}' $, where the
state of \GTIJeHWG{\FyfkKUDd{S}_1', \ldots, \FyfkKUDd{S}_{\gAuWdsUU}'}
is guaranteed to be $ \aOIUkcRb / 2^{\gAuWdsUU} $.
\EfCTdHOd{For all $i$, \GTIJeHWG{\FyfkKUDd{S}_i, \FyfkKUDd{S}_i'}
are supposed to contain \qIsFHsNm{\XhMHOPJK{\IwFBcXWd{q}}}.}
\BeqmjZec accept or reject
\STATE Permute registers $ \GTIJeHWG{\FyfkKUDd{S}_1, \FyfkKUDd{S}_1'},
\ldots, \GTIJeHWG{\FyfkKUDd{S}_{\gAuWdsUU}, \FyfkKUDd{S}_{\gAuWdsUU}'} $
uniformly at random and discard all but
\GTIJeHWG{\FyfkKUDd{S}_1, \FyfkKUDd{S}_1'} and
\GTIJeHWG{\FyfkKUDd{S}_2, \FyfkKUDd{S}_2'}.
\label{GGViMyyh}
\STATE Apply \hqLKZmNM{} on both \GTIJeHWG{\FyfkKUDd{S}_1, \FyfkKUDd{S}_1'}
and \GTIJeHWG{\FyfkKUDd{S}_2, \FyfkKUDd{S}_2'}.
\label{yGAFpbUS}
\STATE Choose $ b \zmKbSHOE \dVhhxuvI{0,1} $ uniformly at random.
\label{yYchijnL}
\IF{$ b = 0 $}
\STATE Apply \cNXPyGJX{} on \GTIJeHWG{\FyfkKUDd{S}_1, \FyfkKUDd{S}_1'}.
\EfCTdHOd{This creates $ \IwFBcXWd{q} \CARSeAIr{0} $
in $\FyfkKUDd{S}_1$. $\FyfkKUDd{S}_1'$ can be discarded.}
\label{UPnMbFeQ}
\STATE Create register \FyfkKUDd{A}, compatible with \LZfegDul,
and initialize its state to \CARSeAIr{\LBGptwLj}.
\STATE Apply \LZfegDul{} on \GTIJeHWG{\FyfkKUDd{P}, \FyfkKUDd{A}}.
\label{ULNGrhTq}
\STATE Apply a phase-flip if both the acceptance qubit
and register $\FyfkKUDd{S}_1$ are $1$.
\EfCTdHOd{This is done by applying the unitary
$ \aOIUkcRb_{\nuoeaWUW{P} \UPMtlSUt \nuoeaWUW{A} \UPMtlSUt \nuoeaWUW{S}_1}
- 2 \ulbwWGDE \UPMtlSUt \fTLcMlGS{1}{1} $ on
\GTIJeHWG{\FyfkKUDd{P}, \FyfkKUDd{A}, \FyfkKUDd{S}_1}.}
\label{WGoOQVbL}
\STATE Apply \XhMHOPJK{\LZfegDul} on \GTIJeHWG{\FyfkKUDd{P}, \FyfkKUDd{A}}.
\label{BSGZETZK}
\STATE Execute \nvBBVayi{alg:post selection} with input
\GTIJeHWG{\FyfkKUDd{S}_2, \FyfkKUDd{S}_2', \FyfkKUDd{S}_1}.
\label{OjQtcmkp}
\IF{\nvBBVayi{alg:post selection} fails}
\RETURN accept
\ENDIF
\STATE Measure \GTIJeHWG{\FyfkKUDd{A}, \FyfkKUDd{S}_2} in
the standard basis.
\label{JXZpLgvW}
\IF{the output of the measurement is \LBGptwLj}
\RETURN reject
\label{KQubmnwO}
\ELSE
\RETURN accept
\ENDIF
\ELSE
\STATE Apply the \KlrOGnZA{} on \GTIJeHWG{\FyfkKUDd{S}_1, \FyfkKUDd{S}_1'}
and \GTIJeHWG{\FyfkKUDd{S}_2, \FyfkKUDd{S}_2'}.
\label{ZrlHrHOA}
\IF{the \KlrOGnZA{} succeeds}
\RETURN accept
\ELSE
\RETURN reject
\ENDIF
\ENDIF
}
Note that \nvBBVayi{alg:QMA_one verifier} runs in polynomial time
and besides performing the circuit \LZfegDul{} and its inverse, it only uses
\eSvaqyjm, \GECLKEJq, \cNXPyGJX, Pauli gates, and classical logical gates.
(This justifies our assumption we made about the gate set
in \BeRtWGMB{oTrRpwia}.)
We have to prove completeness and soundness
in order to prove \XsBWhkhM{GtTrXBeA}.
\cEgZOrfG{UxKSlSFM} proves that in the
honest case \jYpQwXZd{W} always accepts, while
\cEgZOrfG{cNoXknwb} proves that in the dishonest
case \jYpQwXZd{W} rejects with probability at least $ 2^{-52} $.
This shows that $ L \in \mjIKlHwf{1}{1 - 2^{-52}} $.
By \cEgZOrfG{PpcvCrEf},
$ \mjIKlHwf{1}{1 - 2^{-52}} = \hxlhxXVG $
so \XsBWhkhM{GtTrXBeA} follows.
\begin{lemma}[Completeness]
\label{UxKSlSFM}
If $ x \in L $ then the prover can prepare registers
$ \FyfkKUDd{P}, \FyfkKUDd{S}_1, \ldots, \FyfkKUDd{S}_{\gAuWdsUU} $
in such a way that verifier \jYpQwXZd{W} of
\nvBBVayi{alg:QMA_one verifier} accepts with probability $1$.
\end{lemma}
\begin{proof}
Let $ p_x \in \YnJXqMTX{1/2, 1} $ be the maximum probability with which \jYpQwXZd{V}
can be made to accept $x$, where the maximum is taken over
all states in \FyfkKUDd{P}.
Let \[ q \aJAxrBMB \frac{1}{2p} \]
and note that $ q \in \YnJXqMTX{1/2, 1} $.
The honest Merlin prepares \CARSeAIr{\omega_x} in \FyfkKUDd{P},
where \CARSeAIr{\omega_x} is the original witness of
\jYpQwXZd{V} that makes it accept with probability exactly $p_x$.
Furthermore, for all $ i \in \hdzgZeCW{\gAuWdsUU} $, Merlin
applies \XhMHOPJK{\IwFBcXWd{q}} to $\FyfkKUDd{S}_i$.
This creates \qIsFHsNm{\XhMHOPJK{\IwFBcXWd{q}}} in
all \GTIJeHWG{\FyfkKUDd{S}_i, \FyfkKUDd{S}_i'}.
Then Merlin sends registers
$ \FyfkKUDd{P}, \FyfkKUDd{S}_1, \ldots, \FyfkKUDd{S}_{\gAuWdsUU} $
to \jYpQwXZd{W}.
\par
Note that steps~\ref{GGViMyyh}
and \ref{yGAFpbUS} of
\nvBBVayi{alg:QMA_one verifier} don't change
the state because
\[ \qIsFHsNm{\XhMHOPJK{\IwFBcXWd{q}}}
= \sqrt{1 - q} \CARSeAIr{\Phi^{+}} + \XDGqcIBa \sqrt{q} \CARSeAIr{\Psi^{+}}
\in \NxqVzdIP .\]
If, in step~\ref{yYchijnL}, $b$ is chosen to be $1$
then the \KlrOGnZA{} in step~\ref{ZrlHrHOA}
succeeds with certainty, by
\XsBWhkhM{BNEmInIr}.
So, from now on, suppose that $b$ is chosen to be $0$,
in which case we continue to step~\ref{UPnMbFeQ}.
From the arguments of \BeRtWGMB{XoAYconU},
we have that the state of $\FyfkKUDd{S}_1$ after
step~\ref{UPnMbFeQ} is $ \IwFBcXWd{q} \CARSeAIr{0} $.
So the state of \GTIJeHWG{\FyfkKUDd{P}, \FyfkKUDd{A}, \FyfkKUDd{S}_1}
before entering step~\ref{OjQtcmkp} is
\[ \GTIJeHWG{\XhMHOPJK{\LZfegDul} \UPMtlSUt \aOIUkcRb_{\nuoeaWUW{S}_1}}
\GTIJeHWG{\aOIUkcRb - 2 \ulbwWGDE \UPMtlSUt \fTLcMlGS{1}{1}}
\GTIJeHWG{\LZfegDul \UPMtlSUt \IwFBcXWd{q}}
\GTIJeHWG{\CARSeAIr{\omega_x}_{\nuoeaWUW{P}} \UPMtlSUt
\CARSeAIr{\LBGptwLj}_{\nuoeaWUW{A}} \UPMtlSUt \CARSeAIr{0}_{\nuoeaWUW{S}_1}} . \]
We assume that \nvBBVayi{alg:post selection} in
step~\ref{OjQtcmkp} succeeds,
as otherwise we accept.
In this case, by \cEgZOrfG{dUOsaZJF},
the state of \GTIJeHWG{\FyfkKUDd{P}, \FyfkKUDd{A}, \FyfkKUDd{S}_2}
after step~\ref{OjQtcmkp} will be
\[ \GTIJeHWG{\XhMHOPJK{\LZfegDul} \UPMtlSUt \XhMHOPJK{\IwFBcXWd{q}}}
\GTIJeHWG{\aOIUkcRb - 2 \ulbwWGDE \UPMtlSUt \fTLcMlGS{1}{2}}
\GTIJeHWG{\LZfegDul \UPMtlSUt \IwFBcXWd{q}}
\GTIJeHWG{\CARSeAIr{\omega_x}_{\nuoeaWUW{P}} \UPMtlSUt
\CARSeAIr{\LBGptwLj}_{\nuoeaWUW{A}} \UPMtlSUt \CARSeAIr{0}_{\nuoeaWUW{S}_2}} . \]
Let
\begin{align*}
\Delta \aJAxrBMB \aOIUkcRb_{\nuoeaWUW{P}} \UPMtlSUt \KMFbyGRX
\UPMtlSUt \fTLcMlGS{0}{2}
\qquad \text{and} \qquad
\Pi \aJAxrBMB \GTIJeHWG{\XhMHOPJK{\LZfegDul} \ulbwWGDE \LZfegDul} \UPMtlSUt
\GTIJeHWG{\XhMHOPJK{\IwFBcXWd{q}} \fTLcMlGS{1}{2} \IwFBcXWd{q}} .
\end{align*}
Note that the maximum eigenvalue of operator $ \Delta \Pi \Delta $
is $1/2$, with corresponding eigenstate
$ \CARSeAIr{\omega_x}_{\nuoeaWUW{P}} \UPMtlSUt
\CARSeAIr{\LBGptwLj}_{\nuoeaWUW{A} \UPMtlSUt \nuoeaWUW{S}_2} $.
From \cEgZOrfG{bszdFjVr},
\begin{align*}
0 &= \Delta \GTIJeHWG{\aOIUkcRb - 2 \Pi}
\Delta \GTIJeHWG{\CARSeAIr{\omega_x}_{\nuoeaWUW{P}} \UPMtlSUt
\CARSeAIr{\LBGptwLj}_{\nuoeaWUW{A} \UPMtlSUt \nuoeaWUW{S}_2}} \\
&= \GTIJeHWG{\aOIUkcRb_{\nuoeaWUW{P}} \UPMtlSUt
\yAIHQVQC{\LBGptwLj}_{\nuoeaWUW{A} \UPMtlSUt \nuoeaWUW{S}_2}}
\GTIJeHWG{\XhMHOPJK{\LZfegDul} \UPMtlSUt \XhMHOPJK{\IwFBcXWd{q}}}
\GTIJeHWG{\aOIUkcRb - 2 \ulbwWGDE \UPMtlSUt \fTLcMlGS{1}{2}}
\GTIJeHWG{\LZfegDul \UPMtlSUt \IwFBcXWd{q}}
\GTIJeHWG{\CARSeAIr{\omega_x}_{\nuoeaWUW{P}} \UPMtlSUt
\CARSeAIr{\LBGptwLj}_{\nuoeaWUW{A} \UPMtlSUt \nuoeaWUW{S}_2}} .
\end{align*}
It means that the measurement of
step~\ref{JXZpLgvW} will never
output \LBGptwLj.
This finishes the proof of the lemma.
\end{proof}
\newcommand{\wFVdrMmD}{\ensuremath{\rho_1}}
\newcommand{\UttxLjvk}{\ensuremath{\rho_2}}
\newcommand{\RfCClzga}{\ensuremath{\rho_5}}
\newcommand{\VmdniaiO}{\ensuremath{\rho_7}}
\newcommand{\rYGwPCyr}{\ensuremath{\rho_9}}
\begin{lemma}[Soundness]
\label{cNoXknwb}
Let $ x \notin L $ and $n$ sufficiently large.
Suppose that the input to \nvBBVayi{alg:QMA_one verifier}
is \GmcWftaL{} the reduced state on
\GTIJeHWG{\FyfkKUDd{S}_1', \ldots, \FyfkKUDd{S}_{\gAuWdsUU}'}
is $ \aOIUkcRb / 2^{\gAuWdsUU} $.
Then \nvBBVayi{alg:QMA_one verifier} rejects
with probability at least $ 2^{-52} $.
\end{lemma}
\begin{proof}
Let's denote the state of
\GTIJeHWG{\FyfkKUDd{P}, \FyfkKUDd{S}_1, \FyfkKUDd{S}_1',
\FyfkKUDd{S}_2, \FyfkKUDd{S}_2'}, after
step~\ref{GGViMyyh},
by \wFVdrMmD.
\XsBWhkhM{cmMerJQG} implies that
\[ \Vohxuxyh{\FbNPrGDa{\nuoeaWUW{P}}{\wFVdrMmD}}
{\sum_{i=1}^{m} p_i \xi_i \UPMtlSUt \xi_i}
\leq \frac{16}{\gAuWdsUU} . \]
Let's denote the state of the same registers,
after step~\ref{yGAFpbUS},
by \UttxLjvk.
It can be checked by direct calculation that
\begin{align}
\FbNPrGDa{\nuoeaWUW{P} \UPMtlSUt \nuoeaWUW{S}_1 \UPMtlSUt \nuoeaWUW{S}_2}
{\UttxLjvk} = \MrWNkHpB .
\label{NZNouxaY}
\end{align}
From \XsBWhkhM{NQlmRARP}, it holds that
\[ \Vohxuxyh{\FbNPrGDa{\nuoeaWUW{P}}{\UttxLjvk}}
{\sum_{i=1}^{m} p_i \sigma_i \UPMtlSUt \sigma_i}
\leq \frac{16}{\gAuWdsUU} \text{,} \]
where $ \sigma_i \aJAxrBMB \mpAYNUjM{\xi_i} $.
By \cEgZOrfG{CpyQcgif},
there exists a $\UttxLjvk'$ \GmcWftaL{}
\begin{align}
\FbNPrGDa{\nuoeaWUW{P}}{\UttxLjvk'}
&= \sum_{i=1}^{m} p_i \sigma_i \UPMtlSUt \sigma_i \nonumber
\intertext{and}
\Vohxuxyh{\UttxLjvk}{\UttxLjvk'}
&\leq \sqrt{\frac{32}{\gAuWdsUU}} .
\label{RcJveNBK}
\end{align}
Let us suppose, from now on, that before entering
step~\ref{yYchijnL} the state of the
system is $\UttxLjvk'$.
This will result in a bias of at most
$ \sqrt{32 / \gAuWdsUU} $ in the trace distance in
the rest of the states that we calculate.
Throughout the rest of the proof, we will assume that
the \KlrOGnZA{} on input \FbNPrGDa{\nuoeaWUW{P}}{\UttxLjvk'}
rejects with probability at most
$ \varepsilon \aJAxrBMB 2 \cdot 2^{-52} + \sqrt{32 / \gAuWdsUU} = 2^{-50}$,
as otherwise we are done with the proof.
With this in mind, the rest of the proof will
only deal with the case when $b$ is chosen to be $0$ in
step~\ref{yYchijnL}.
In this case we continue to
step~\ref{UPnMbFeQ}.
With these assumptions, the state of the system
after step~\ref{UPnMbFeQ} is
\[ \RfCClzga \aJAxrBMB \GTIJeHWG{\cNXPyGJX \UPMtlSUt
\aOIUkcRb_{\nuoeaWUW{P} \UPMtlSUt \nuoeaWUW{S}_2 \UPMtlSUt \nuoeaWUW{S}_2'}}
\UttxLjvk' \GTIJeHWG{\XhMHOPJK{\cNXPyGJX} \UPMtlSUt \aOIUkcRb_{\nuoeaWUW{P}
\UPMtlSUt \nuoeaWUW{S}_2 \UPMtlSUt \nuoeaWUW{S}_2'}} . \]
Let's denote the state of the whole system after
step~\ref{ULNGrhTq} by
\[ \VmdniaiO \aJAxrBMB \GTIJeHWG{\LZfegDul \UPMtlSUt
\aOIUkcRb_{\nuoeaWUW{S}_1 \UPMtlSUt \nuoeaWUW{S}_1'
\UPMtlSUt \nuoeaWUW{S}_2 \UPMtlSUt \nuoeaWUW{S}_2'}}
\GTIJeHWG{\KMFbyGRX \UPMtlSUt \RfCClzga}
\GTIJeHWG{\XhMHOPJK{\LZfegDul} \UPMtlSUt
\aOIUkcRb_{\nuoeaWUW{S}_1 \UPMtlSUt \nuoeaWUW{S}_1'
\UPMtlSUt \nuoeaWUW{S}_2 \UPMtlSUt \nuoeaWUW{S}_2'}}. \]
Since the acceptance probability
of \LZfegDul{} is at most $ 4^{-n} $, we have that
\[ \gyBwVFHt{\VmdniaiO \CGvPMhCL} \leq \frac{1}{4^n} \text{,} \]
where $ \CGvPMhCL \aJAxrBMB \ulbwWGDE \UPMtlSUt
\aOIUkcRb_{\nuoeaWUW{S}_1 \UPMtlSUt \nuoeaWUW{S}_1' \UPMtlSUt
\nuoeaWUW{S}_2 \UPMtlSUt \nuoeaWUW{S}_2'}$.
Let $\VmdniaiO'$ be the projection
of $\VmdniaiO$ to the rejection subspace, \YRlgvUEY{}
\[ \VmdniaiO' \aJAxrBMB
\frac{\dIKLUvxq \VmdniaiO \dIKLUvxq}
{\gyBwVFHt{\VmdniaiO \dIKLUvxq}} . \]
From \cEgZOrfG{kvwuafiV} and
\XsBWhkhM{gYublpMS}, we have that
\[ 1 - \frac{1}{4^n} \leq \DGaDmHJe{\VmdniaiO}{\VmdniaiO'}^2
\leq 1 - \Vohxuxyh{\VmdniaiO}{\VmdniaiO'}^2 \]
from which it follows that
\[ \Vohxuxyh{\VmdniaiO}{\VmdniaiO'} \leq \frac{1}{2^n} . \]
Now suppose that before entering
step~\ref{WGoOQVbL} the state of the
system is $\VmdniaiO'$ instead of $\VmdniaiO$.
This will result in an additional bias of at most
$ 2^{-n} $ in the trace distance in
the rest of the states that we calculate.
Since $\VmdniaiO'$ lies in the rejection subspace,
\[ \GTIJeHWG{\GTIJeHWG{\aOIUkcRb - 2 \ulbwWGDE \UPMtlSUt \fTLcMlGS{1}{1}} \UPMtlSUt
\aOIUkcRb_{\nuoeaWUW{S}_1' \UPMtlSUt \nuoeaWUW{S}_2 \UPMtlSUt \nuoeaWUW{S}_2'}}
\VmdniaiO'
\GTIJeHWG{\GTIJeHWG{\aOIUkcRb - 2 \ulbwWGDE \UPMtlSUt \fTLcMlGS{1}{1}} \UPMtlSUt
\aOIUkcRb_{\nuoeaWUW{S}_1' \UPMtlSUt \nuoeaWUW{S}_2 \UPMtlSUt \nuoeaWUW{S}_2'}}
= \VmdniaiO' \text{,} \]
which means that step~\ref{WGoOQVbL}
doesn't change the state.
So the state of the system before entering
step~\ref{BSGZETZK}
is $\VmdniaiO'$.
Let us change the state again, at this time
from $\VmdniaiO'$ back to $\VmdniaiO$.
This will result in another bias of
at most $ 2^{-n} $.
If the state of the system is $\VmdniaiO$
before entering step~\ref{BSGZETZK}
then the state after step~\ref{BSGZETZK}
will be
\[ \GTIJeHWG{\XhMHOPJK{\LZfegDul} \UPMtlSUt
\aOIUkcRb_{\nuoeaWUW{S}_1 \UPMtlSUt \nuoeaWUW{S}_1'
\UPMtlSUt \nuoeaWUW{S}_2 \UPMtlSUt \nuoeaWUW{S}_2'}}
\VmdniaiO
\GTIJeHWG{\LZfegDul \UPMtlSUt
\aOIUkcRb_{\nuoeaWUW{S}_1 \UPMtlSUt \nuoeaWUW{S}_1'
\UPMtlSUt \nuoeaWUW{S}_2 \UPMtlSUt \nuoeaWUW{S}_2'}}
= \KMFbyGRX \UPMtlSUt \RfCClzga . \]
\par
From \cEgZOrfG{ijiuOsQF},
together with the assumption we made about the
success probability of the \KlrOGnZA, we get that
there exists a set of states
\LtfmKMob{\CARSeAIr{\varphi_i}}{\CARSeAIr{\varphi_i} \in \NxqVzdIP
\text{ or } \CARSeAIr{\varphi_i} \in \NticoQpX}
\GmcWftaL{}
\begin{align}
\Vohxuxyh{\FbNPrGDa{\nuoeaWUW{P}}{\UttxLjvk'}}
{\sum_{i=1}^m p_i \GTIJeHWG{\yAIHQVQC{\varphi_i}}^{\UPMtlSUt 2}}
\leq 6 \sqrt{\varepsilon} .
\label{kLpPiynP}
\end{align}
This implies that
\[ \Vohxuxyh{\FbNPrGDa{\nuoeaWUW{P}}{\RfCClzga}}
{\rYGwPCyr} \leq 6 \sqrt{\varepsilon} \text{,} \]
where
\[ \rYGwPCyr \aJAxrBMB \sum_{i=1}^m p_i
\GTIJeHWG{\cNXPyGJX \yAIHQVQC{\varphi_i} \XhMHOPJK{\cNXPyGJX}}
\UPMtlSUt \yAIHQVQC{\varphi_i}
\text{, \qquad}
\rYGwPCyr \in \sYllRVEw{\nuoeaWUW{S}_1 \UPMtlSUt \nuoeaWUW{S}_1'
\UPMtlSUt \nuoeaWUW{S}_2 \UPMtlSUt \nuoeaWUW{S}_2'} . \]
Now let us change the state of
\GTIJeHWG{\FyfkKUDd{S}_1, \FyfkKUDd{S}_1', \FyfkKUDd{S}_2, \FyfkKUDd{S}_2'}
from \FbNPrGDa{\nuoeaWUW{P}}{\RfCClzga}
to \rYGwPCyr.
This will result in another bias of at most $ 6 \sqrt{\varepsilon} $.
(Note that \FyfkKUDd{P} is not touched by the
algorithm after step~\ref{BSGZETZK},
so we don't keep track of its state.)
From Eqs.~\eqref{NZNouxaY},
\eqref{RcJveNBK},
and \eqref{kLpPiynP}, it follows that
\[ \Vohxuxyh{\FbNPrGDa{\nuoeaWUW{S}_1 \UPMtlSUt \nuoeaWUW{S}_2}
{\sum_{i=1}^m p_i \GTIJeHWG{\yAIHQVQC{\varphi_i}}^{\UPMtlSUt 2}}}
{\MrWNkHpB}
\leq \sqrt{\frac{32}{\gAuWdsUU}} + 6 \sqrt{\varepsilon}
< \frac{1}{8} . \]
So \rYGwPCyr{} satisfies the requirements of
\cEgZOrfG{oBbNiOJY} below.
This means that \nvBBVayi{alg:post selection} in
step~\ref{OjQtcmkp} succeeds
with probability at least $2^{-20}$, in which case we continue
to step~\ref{JXZpLgvW}.
\par
We now argue that, conditioned on \nvBBVayi{alg:post selection}
being successful, the measurement in
step~\ref{JXZpLgvW} outputs \LBGptwLj{}
with certainty.
This will finish the proof.
Note that \nvBBVayi{alg:post selection} can't change the
state of \FyfkKUDd{A} as it was independent of
\GTIJeHWG{\FyfkKUDd{S}_2, \FyfkKUDd{S}_2', \FyfkKUDd{S}_1}
before executing \nvBBVayi{alg:post selection}.
So before entering step~\ref{JXZpLgvW},
the state of \FyfkKUDd{A} is still \CARSeAIr{\LBGptwLj}.
Now we argue that after successfully executing
\nvBBVayi{alg:post selection}, the state of
$\FyfkKUDd{S}_2$ will be \CARSeAIr{0}.
Let us take some $ \CARSeAIr{\varphi} \in
\nuoeaWUW{S}_1 \UPMtlSUt \nuoeaWUW{S}_1' $ that
belongs to either \NxqVzdIP{} or \NticoQpX.
Here we only argue about the case when
$ \CARSeAIr{\varphi} \in \NxqVzdIP $ as
the other case can be proven by exactly the same way.
We can write \CARSeAIr{\varphi} as
\[ \CARSeAIr{\varphi} = a \CARSeAIr{\Phi^{+}}
+ b \CARSeAIr{\Psi^{+}}, \qquad
a, b \in \ynpHhbyi, \qquad
\UylNyihM{a}^2 + \UylNyihM{b}^2 = 1 . \]
It is easy to see that after applying \cNXPyGJX{}
to \CARSeAIr{\varphi}, the resulting state on $\FyfkKUDd{S}_1$ will be
$ a \CARSeAIr{0} - b \CARSeAIr{1} $.
Suppose that the state of \GTIJeHWG{\FyfkKUDd{S}_2, \FyfkKUDd{S}_2'}
is \CARSeAIr{\varphi} and the state of $\FyfkKUDd{S}_1$
is $ a \CARSeAIr{0} - b \CARSeAIr{1} $.
It can be shown by direct calculation that
\[ \GTIJeHWG{\fTLcMlGS{1}{2} \UPMtlSUt
\yAIHQVQC{\Phi^{+}}_{\nuoeaWUW{S}_2' \UPMtlSUt \nuoeaWUW{S}_1}}
\CARSeAIr{\varphi} \UPMtlSUt \GTIJeHWG{a \CARSeAIr{0} - b \CARSeAIr{1}}
= 0 . \]
This means that if \nvBBVayi{alg:post selection} is executed
with the above input and the measurement in the
algorithm results in \CARSeAIr{\Phi^{+}}, then
the state of $\FyfkKUDd{S}_2$ will be \CARSeAIr{0}.
Similarly to the above, it can also be shown that
\[ \GTIJeHWG{\fTLcMlGS{0}{2} \UPMtlSUt
\yAIHQVQC{\Psi^{+}}_{\nuoeaWUW{S}_2' \UPMtlSUt \nuoeaWUW{S}_1}}
\CARSeAIr{\varphi} \UPMtlSUt \GTIJeHWG{a \CARSeAIr{0} - b \CARSeAIr{1}}
= 0 . \]
This means that if the measurement in \nvBBVayi{alg:post selection}
results in \CARSeAIr{\Psi^{+}} then the state of
$\FyfkKUDd{S}_2$ will be \CARSeAIr{1}.
In this case, \nvBBVayi{alg:post selection}
applies \FqTFSrwh{} on $\FyfkKUDd{S}_2$ so
the state of this register, after the algorithm,
will be \CARSeAIr{0}.
Since \rYGwPCyr{} is a convex combination of states of the
above form, we got that if the state of
\GTIJeHWG{\FyfkKUDd{S}_1, \FyfkKUDd{S}_1', \FyfkKUDd{S}_2, \FyfkKUDd{S}_2'}
is \rYGwPCyr, before entering step~\ref{OjQtcmkp},
then \nvBBVayi{alg:post selection}
succeeds with probability at least $2^{-20}$
and, conditioned on success, \nvBBVayi{alg:QMA_one verifier}
rejects in step~\ref{KQubmnwO}
with certainty.
\par
However, we did modify the state during our
analysis four times, so we have to account for the bias
they caused, which is at most
\[ \frac{1}{2^{n-1}} + \sqrt{\frac{32}{\gAuWdsUU}} + 6 \sqrt{\varepsilon} . \]
So the real rejection probability, with the
original input, is at least
\[ \frac{1}{2^{20}} - \GTIJeHWG{\frac{1}{2^{n-1}} + \sqrt{\frac{32}{\gAuWdsUU}}
+ 6 \sqrt{\varepsilon}}
= \frac{1}{2^{21}} - \frac{1}{2^{n-1}}
\geq \frac{1}{2^{22}} \text{,} \]
where the last inequality is true for
$ n \geq 23 $.
\end{proof}
\begin{lemma}
\label{oBbNiOJY}
Suppose that before entering step~\ref{OjQtcmkp}
of \nvBBVayi{alg:QMA_one verifier}, the state of
\GTIJeHWG{\FyfkKUDd{S}_1, \FyfkKUDd{S}_1', \FyfkKUDd{S}_2, \FyfkKUDd{S}_2'} is
\begin{align*}
\rho \aJAxrBMB \sum_{i=1}^m p_i \GTIJeHWG{\FZGQkLzf}
\UPMtlSUt \sigma_i \text{,}
\end{align*}
for some $ m \in \FLbfwYzZ $, probability distribution
\LtfmKMob{p_i}{i = 1, \ldots, m}, and states
$ \sigma_i \in \sYllRVEw{\nuoeaWUW{S}_2 \UPMtlSUt \nuoeaWUW{S}_2'}
\cong \sYllRVEw{\nuoeaWUW{S}_1 \UPMtlSUt \nuoeaWUW{S}_1'} $.
Further assume that
\begin{align}
\Vohxuxyh{\FbNPrGDa{\nuoeaWUW{S}_1 \UPMtlSUt \nuoeaWUW{S}_2}
{\sum_{i=1}^m p_i \sigma_i \UPMtlSUt \sigma_i}}
{\MrWNkHpB} \leq \frac{1}{8} .
\label{IMZsmswK}
\end{align}
Then \nvBBVayi{alg:post selection}, in
step~\ref{OjQtcmkp}, will succeed
with probability at least $2^{-20}$.
\end{lemma}
The idea behind the proof of \cEgZOrfG{oBbNiOJY}
is very simple.
We show that if the measurement in \nvBBVayi{alg:post selection}
fails with high probability on a state of the form
$ \CQiPcAHs \UPMtlSUt \zeta $, where
$ \zeta \in \sYllRVEw{\nuoeaWUW{S}_1} $ is an arbitrary state,
then \CQiPcAHs{} must be close to either
\CARSeAIr{\MEVIdGJb} or \CARSeAIr{\VPzthEpq}.
But then the convex combination of the states
$ \FbNPrGDa{\nuoeaWUW{S}_1}{\sigma_i} \UPMtlSUt \CQiPcAHs $
won't be close to the maximally mixed state.
\begin{proof}[Proof of \cEgZOrfG{oBbNiOJY}]
Let us group the states in ensemble $\rho$ with respect to
their reduced state on $\nuoeaWUW{S}_2'$ being close to
\CARSeAIr{\MEVIdGJb}, or to \CARSeAIr{\VPzthEpq}, or being far
from both.
Formally, let $ \varepsilon_1 \aJAxrBMB 2^{-9} $,
\begin{align*}
\rkQwwcQs &\aJAxrBMB \LtfmKMob{i}{1 \leq i \leq m, \VslhEMHJ
\Vohxuxyh{\CQiPcAHs}{\yAIHQVQC{\MEVIdGJb}}
\leq \varepsilon_1} \text{,} \\
\tTkICuaj &\aJAxrBMB \LtfmKMob{i}{1 \leq i \leq m, \VslhEMHJ
\Vohxuxyh{\CQiPcAHs}{\yAIHQVQC{\VPzthEpq}}
\leq \varepsilon_1} \text{,} \\
\fTIaZVxF &\aJAxrBMB \hdzgZeCW{m} \setminus \GTIJeHWG{\rkQwwcQs \cup \tTkICuaj} .
\end{align*}
Since $ \Vohxuxyh{\CARSeAIr{\MEVIdGJb}}{\CARSeAIr{\VPzthEpq}} = 1 $
and $ \varepsilon_1 < 1/2 $, from the triangle
inequality we have that $ \rkQwwcQs \cap \tTkICuaj = \emptyset $.
\par
We first show that if the probability of \fTIaZVxF{}
is at least $ \varepsilon_2 \aJAxrBMB 1/4 $
then we are done.
So assume for now that
$ \varepsilon_2 \leq \sum_{i \in \fTIaZVxF} p_i $.
For all $ i \in \fTIaZVxF $ we have that
\begin{align}
\sqrt{\bymhjSLm{\MEVIdGJb} \CQiPcAHs \CARSeAIr{\MEVIdGJb}}
&= \DGaDmHJe{\CQiPcAHs}{\yAIHQVQC{\MEVIdGJb}}
\label{QikgnKvj} \\
&\leq \sqrt{1 - \Vohxuxyh{\CQiPcAHs}
{\yAIHQVQC{\MEVIdGJb}}^2}
\label{qTeweEEg} \\
&< \sqrt{1 - \varepsilon_1^2} \text{,}
\label{pEqSxNzP}
\end{align}
where \eqref{QikgnKvj} follows from
\eqref{ZIjFQVNH}, \eqref{qTeweEEg}
follows from \XsBWhkhM{gYublpMS},
and \eqref{pEqSxNzP} is from the definition of \fTIaZVxF.
The above implies that
\begin{align*}
\bymhjSLm{\MEVIdGJb} \CQiPcAHs \CARSeAIr{\MEVIdGJb}
< 1 - \varepsilon_1^2
\qquad \text{and similarly} \qquad
\bymhjSLm{\VPzthEpq} \CQiPcAHs \CARSeAIr{\VPzthEpq}
< 1 - \varepsilon_1^2 .
\end{align*}
From the above and using the fact that
\[ \bymhjSLm{\MEVIdGJb} \CQiPcAHs \CARSeAIr{\MEVIdGJb}
+ \bymhjSLm{\VPzthEpq} \CQiPcAHs \CARSeAIr{\VPzthEpq}
= \gyBwVFHt{\CQiPcAHs}
= 1 \text{,} \]
we get that
\begin{align*}
\varepsilon_1^2 < \bymhjSLm{\MEVIdGJb} \CQiPcAHs \CARSeAIr{\MEVIdGJb}
\qquad \text{and} \qquad
\varepsilon_1^2 < \bymhjSLm{\VPzthEpq} \CQiPcAHs \CARSeAIr{\VPzthEpq} .
\end{align*}
Let us take an arbitrary state
\[ \CARSeAIr{\psi} \aJAxrBMB a \CARSeAIr{\MEVIdGJb}
+ b \CARSeAIr{\VPzthEpq} \in \nuoeaWUW{S}_1,
\qquad a, b \in \ynpHhbyi, \qquad
\UylNyihM{a}^2 + \UylNyihM{b}^2 = 1 . \]
If the state of \GTIJeHWG{\FyfkKUDd{S}_2', \FyfkKUDd{S}_1},
in the input to \nvBBVayi{alg:post selection}, is
$ \CQiPcAHs \UPMtlSUt \yAIHQVQC{\psi} $ then
the algorithm will succeed with probability
\begin{align*}
\gyBwVFHt{\GTIJeHWG{\CQiPcAHs \UPMtlSUt \yAIHQVQC{\psi}}
\KlqDGJMb}
&= \UylNyihM{a}^2 \cdot \bymhjSLm{\MEVIdGJb} \CQiPcAHs \CARSeAIr{\MEVIdGJb}
+ \UylNyihM{b}^2 \cdot \bymhjSLm{\VPzthEpq} \CQiPcAHs \CARSeAIr{\VPzthEpq} \\
&> \varepsilon_1^2 \GTIJeHWG{\UylNyihM{a}^2 + \UylNyihM{b}^2} \\
&= \varepsilon_1^2 \text{,}
\end{align*}
where the first equality follows from direct calculation using
\begin{align*}
\CARSeAIr{\Phi^{+}} = \frac{\CARSeAIr{\MEVIdGJb} \UPMtlSUt \CARSeAIr{\MEVIdGJb}
+ \CARSeAIr{\VPzthEpq} \UPMtlSUt \CARSeAIr{\VPzthEpq}}{\sqrt{2}}
\qquad \text{and} \qquad
\CARSeAIr{\Psi^{+}} = \frac{\CARSeAIr{\MEVIdGJb} \UPMtlSUt \CARSeAIr{\MEVIdGJb}
- \CARSeAIr{\VPzthEpq} \UPMtlSUt \CARSeAIr{\VPzthEpq}}{\sqrt{2}} .
\end{align*}
This implies that if the state of
\GTIJeHWG{\FyfkKUDd{S}_2', \FyfkKUDd{S}_1} is
$ \CQiPcAHs \UPMtlSUt \zeta $, for any
$ \zeta \in \sYllRVEw{\nuoeaWUW{S}_1} $, then the
probability that \nvBBVayi{alg:post selection} succeeds
is at least $\varepsilon_1^2$.
We got that if $ \varepsilon_2 \leq \sum_{i \in \fTIaZVxF} p_i $
then \nvBBVayi{alg:post selection} succeeds with probability
at least $ \varepsilon_1^2 \varepsilon_2 = 2^{-20} $,
in which case we are done.
\par
So, from now on, assume that
$ \sum_{i \in \fTIaZVxF} p_i < \varepsilon_2 $.
We will show that this assumption leads to a contradiction,
which will finish the proof.
\cEgZOrfG{CpyQcgif} implies that
\begin{align*}
\forall i \in \rkQwwcQs, \VslhEMHJ
\exists \tau_i \in \sYllRVEw{\nuoeaWUW{S}_2}
&\lyfuwjUX \Vohxuxyh{\sigma_i}
{\tau_i \UPMtlSUt \yAIHQVQC{\MEVIdGJb}}
\leq \sqrt{2 \varepsilon_1} \text{,} \\
\forall i \in \tTkICuaj, \VslhEMHJ
\exists \tau_i \in \sYllRVEw{\nuoeaWUW{S}_2}
&\lyfuwjUX \Vohxuxyh{\sigma_i}
{\tau_i \UPMtlSUt \yAIHQVQC{\VPzthEpq}}
\leq \sqrt{2 \varepsilon_1} .
\end{align*}
We now replace $\sigma_i$ with
$ \tau_i \UPMtlSUt \yAIHQVQC{\MEVIdGJb} $ or
$ \tau_i \UPMtlSUt \yAIHQVQC{\VPzthEpq} $ in $\rho$.
Formally, let us define
\begin{align*}
\tVRtKpaI &\aJAxrBMB \sum_{i \in \fTIaZVxF} p_i
\GTIJeHWG{\FZGQkLzf} \UPMtlSUt \sigma_i \text{,} \\
\rho' &\aJAxrBMB \sum_{i \in \rkQwwcQs} p_i
\GTIJeHWG{\cNXPyGJX \GTIJeHWG{\tau_i \UPMtlSUt \yAIHQVQC{\MEVIdGJb}}
\XhMHOPJK{\cNXPyGJX}} \UPMtlSUt \tau_i \UPMtlSUt \yAIHQVQC{\MEVIdGJb} \\
&\qquad {} + \sum_{i \in \tTkICuaj} p_i
\GTIJeHWG{\cNXPyGJX \GTIJeHWG{\tau_i \UPMtlSUt \yAIHQVQC{\VPzthEpq}}
\XhMHOPJK{\cNXPyGJX}} \UPMtlSUt \tau_i \UPMtlSUt \yAIHQVQC{\VPzthEpq} \\
&\qquad {} + \tVRtKpaI \text{,}
\end{align*}
where $ \gyBwVFHt{\tVRtKpaI} < \varepsilon_2 $.
Note that $ \Vohxuxyh{\rho}{\rho'}
< 2 \sqrt{2 \varepsilon_1} $, which, together
with \eqref{IMZsmswK}, implies that
\begin{align}
\Vohxuxyh{\LvsCZiIr}{\MrWNkHpB}
\leq 2 \sqrt{2 \varepsilon_1} + \frac{1}{8}
= \frac{1}{4} \text{,}
\label{HAaUNije}
\end{align}
where
\[ \LvsCZiIr \aJAxrBMB
\FbNPrGDa{\nuoeaWUW{S}_1 \UPMtlSUt \nuoeaWUW{S}_2}
{\GTIJeHWG{\XhMHOPJK{\cNXPyGJX} \UPMtlSUt \aOIUkcRb_{\nuoeaWUW{S}_2 \UPMtlSUt \nuoeaWUW{S}_2'}}
\rho' \GTIJeHWG{\cNXPyGJX \UPMtlSUt
\aOIUkcRb_{\nuoeaWUW{S}_2 \UPMtlSUt \nuoeaWUW{S}_2'}}} . \]
On the other hand, we have that
\begin{align*}
\LvsCZiIr = \hxCscLfe + \RWgRdDui \text{,}
\end{align*}
for some \RWgRdDui, where we used the shorthand
$ p_{+} \aJAxrBMB \sum_{i \in \rkQwwcQs} p_i $ and
$ p_{-} \aJAxrBMB \sum_{i \in \tTkICuaj} p_i $.
Note that $ \gyBwVFHt{\RWgRdDui} < \varepsilon_2 $, so
\cEgZOrfG{XViefLtt} implies that
\begin{align}
\Vohxuxyh{\LvsCZiIr}{\hxCscLfe}
\leq \frac{\varepsilon_2}{2} .
\label{FoAMkNPk}
\end{align}
The following calculation will lead us to a contradiction.
\begin{align}
\frac{1}{2} &\leq \frac{1}{2}
\GTIJeHWG{\UylNyihM{\frac{1}{4} - p_{+}} + \UylNyihM{\frac{1}{4} - p_{-}}
+ \frac{1}{2}} \nonumber \\
&= \frac{1}{2} \JjPvCTWX{\MrWNkHpB - \GTIJeHWG{\hxCscLfe}}
\label{qEIQMTWs} \\
&= \Vohxuxyh{\MrWNkHpB}{\hxCscLfe} \nonumber \\
&\leq \Vohxuxyh{\LvsCZiIr}{\MrWNkHpB}
+ \Vohxuxyh{\LvsCZiIr}{\hxCscLfe}
\label{shbenJPg} \\
&\leq \Vohxuxyh{\LvsCZiIr}{\MrWNkHpB}
+ \frac{\varepsilon_2}{2} \text{,}
\label{jdtzJeZs}
\end{align}
where \eqref{qEIQMTWs} is because
the eigenvalues of $ \MrWNkHpB - \GTIJeHWG{\hxCscLfe} $ are
$ \frac{1}{4} - p_{+} $, $ \frac{1}{4} - p_{-} $,
and $ \frac{1}{4} $ with multiplicity $2$.
Eq.~\eqref{shbenJPg} follows from
the triangle inequality and at \eqref{jdtzJeZs}
we used \eqref{FoAMkNPk}.
Eq.~\eqref{jdtzJeZs} implies that
\[ \Vohxuxyh{\LvsCZiIr}{\MrWNkHpB}
\geq \frac{1}{2} - \frac{\varepsilon_2}{2}
= \frac{3}{8} \text{,} \]
which contradicts to \eqref{HAaUNije}.
So we conclude that it must be that
$ \varepsilon_2 \leq \sum_{i \in \fTIaZVxF} p_i $,
in which case \nvBBVayi{alg:post selection} succeeds with
the desired probability, as argued above.
\end{proof}
The following lemma is similar to Proposition~24
of \cite{Kobayashi2013}.
\begin{lemma}
\label{ijiuOsQF}
Let $\FyfkKUDd{S}_1$, $\FyfkKUDd{S}_1'$, $\FyfkKUDd{S}_2$,
$\FyfkKUDd{S}_2'$ be single-qubit registers and
let the state of \GTIJeHWG{\FyfkKUDd{S}_1, \FyfkKUDd{S}_1',
\FyfkKUDd{S}_2, \FyfkKUDd{S}_2'} be
\[ \rho \aJAxrBMB \sum_{i=1}^m p_i \sigma_i
\UPMtlSUt \sigma_i \text{,} \]
where $ m \in \FLbfwYzZ $, \LtfmKMob{p_i}{i = 1, \ldots, m}
is a probability distribution, and
$ \sigma_i = \mpAYNUjM{\xi_i} $, for some
$ \xi_i \in \sYllRVEw{\nuoeaWUW{S}_1 \UPMtlSUt \nuoeaWUW{S}_1'}
\cong \sYllRVEw{\nuoeaWUW{S}_2 \UPMtlSUt \nuoeaWUW{S}_2'} $.
Let $ 0 \leq \varepsilon < 1 $.
If the \KlrOGnZA{}, applied between \GTIJeHWG{\FyfkKUDd{S}_1, \FyfkKUDd{S}_1'}
and \GTIJeHWG{\FyfkKUDd{S}_2, \FyfkKUDd{S}_2'}, succeeds
with probability at least $ 1 - \varepsilon $ then
there exist a set of states
\[ \LtfmKMob{\CARSeAIr{\varphi_i}}{1 \leq i \leq m, \VslhEMHJ
\CARSeAIr{\varphi_i} \in \NxqVzdIP \text{ or }
\CARSeAIr{\varphi_i} \in \NticoQpX} \]
\GmcWftaL
\begin{align*}
\Vohxuxyh{\rho}{\sum_{i=1}^m p_i
\yAIHQVQC{\varphi_i} \UPMtlSUt \yAIHQVQC{\varphi_i}}
\leq 6 \sqrt{\varepsilon} .
\end{align*}
\end{lemma}
\begin{proof}
On input $ \sigma_i \UPMtlSUt \sigma_i $ the
\KlrOGnZA{} succeeds with probability
$ \GTIJeHWG{1 + \gyBwVFHt{\sigma_i^2}} / 2 $, by
\XsBWhkhM{BNEmInIr}.
So with input $\rho$ the \KlrOGnZA{} succeeds with
probability
\[ \sum_{i=1}^m p_i \frac{1 + \gyBwVFHt{\sigma_i^2}}{2}
\geq 1 - \varepsilon . \]
If $ \varepsilon = 0 $ it implies that
all $\sigma_i$'s are pure and the statement of the
lemma follows.
So, from now on, assume that $ 0 < \varepsilon $.
Then the above inequality intuitively means that
for most of the $i$'s, \gyBwVFHt{\sigma_i^2} must be
close to $1$.
Formally, let
\begin{align*}
\fTIaZVxF &\aJAxrBMB \LtfmKMob{i}{1 \leq i \leq m, \VslhEMHJ
\gyBwVFHt{\sigma_i^2} \leq 1 - 2 \sqrt{\varepsilon}}
\text{,} \\
\UuexCuFw &\aJAxrBMB \hdzgZeCW{m} \setminus \fTIaZVxF .
\end{align*}
Suppose towards contradiction that
$ 2 \sqrt{\varepsilon} \leq \sum_{i \in \fTIaZVxF} p_i $.
Then the probability that the \KlrOGnZA{} fails is
\begin{align*}
\sum_{i=1}^m p_i \frac{1 - \gyBwVFHt{\sigma_i^2}}{2}
&\geq \sum_{i \in \fTIaZVxF} p_i \frac{1 - \gyBwVFHt{\sigma_i^2}}{2} \\
&\geq \sum_{i \in \fTIaZVxF} p_i \frac{1 -
\GTIJeHWG{1 - 2 \sqrt{\varepsilon}}}{2} \\
&\geq \sqrt{\varepsilon} \cdot \sum_{i \in \fTIaZVxF} p_i \\
&\geq 2 \varepsilon \text{,}
\end{align*}
which is a contradiction.
This implies that
$ \sum_{i \in \fTIaZVxF} p_i < 2 \sqrt{\varepsilon} $.
For all $ i \in \UuexCuFw $, let $\lambda_i$ be the maximum
eigenvalue of $\sigma_i$ and \CARSeAIr{\varphi_i} be the
corresponding eigenstate.
Note that either $ \CARSeAIr{\varphi_i} \in \NxqVzdIP $
or $ \CARSeAIr{\varphi_i} \in \NticoQpX $.
From the definition of \UuexCuFw, we have that
\begin{align*}
1 - 2 \sqrt{\varepsilon}
< \gyBwVFHt{\sigma_i^2}
\leq \JjPvCTWX{\sigma_i} \cdot \bjhdjWcQ{\sigma_i}
= \bjhdjWcQ{\sigma_i}
= \lambda_i \text{,}
\end{align*}
where the second inequality follows
from \cEgZOrfG{qtWTtRce}.
The above calculation, together with
\cEgZOrfG{eJTZpBKB},
imply that
\begin{align}
\forall i \in \UuexCuFw \lyfuwjUX
\Vohxuxyh{\sigma_i}{\yAIHQVQC{\varphi_i}}
\leq 2 \sqrt{\varepsilon} .
\label{ZgYnKQXO}
\end{align}
We can now bound the required trace distance.
\begin{align}
\Vohxuxyh{\rho}{\sum_{i=1}^m p_i
\GTIJeHWG{\yAIHQVQC{\varphi_i}}^{\UPMtlSUt 2}}
&\leq \Vohxuxyh{\sum_{i=1}^m p_i
\sigma_i^{\UPMtlSUt 2}}
{\sum_{i \in \UuexCuFw} p_i \sigma_i^{\UPMtlSUt 2}}
+ \Vohxuxyh{\sum_{i \in \UuexCuFw} p_i
\sigma_i^{\UPMtlSUt 2}}
{\sum_{i \in \UuexCuFw} p_i
\GTIJeHWG{\yAIHQVQC{\varphi_i}}^{\UPMtlSUt 2}} \nonumber \\
&\qquad {} + \Vohxuxyh{\sum_{i \in \UuexCuFw} p_i
\GTIJeHWG{\yAIHQVQC{\varphi_i}}^{\UPMtlSUt 2}}
{\sum_{i=1}^m p_i \GTIJeHWG{\yAIHQVQC{\varphi_i}}^{\UPMtlSUt 2}}
\label{cMqMmEBg} \\
&\leq 6 \sqrt{\varepsilon} \text{,}
\label{TCKJspSw}
\end{align}
where \eqref{cMqMmEBg} follows from
the triangle inequality and at
\eqref{TCKJspSw} we used
\cEgZOrfG{XViefLtt} twice and
\eqref{ZgYnKQXO}.
\end{proof}
\section*{Acknowledgements}
The author would like to thank Rahul Jain,
Sarvagya Upadhyay, and Penghui Yao for helpful
discussions on the topic.
\fERcsqXi

\newcommand{\etalchar}[1]{$^{#1}$}
\begin{thebibliography}{BCWdW01}

\bibitem[Aar09]{Aaronson2009a}
Scott Aaronson.
\newblock On perfect completeness for {QMA}.
\newblock {\em Quantum Information and Computation}, 9(1):81--89, January 2009,
  \href{http://arxiv.org/abs/0806.0450}
  {\textsc{arXiv:\linebreak[0]0806.0450}}.

\bibitem[AB09]{Arora2009}
Sanjeev Arora and Boaz Barak.
\newblock {\em Computational Complexity: A Modern Approach}.
\newblock Cambridge University Press, New York, NY, USA, 1st edition, 2009.

\bibitem[ABD{\etalchar{+}}09]{Aaronson2009}
Scott Aaronson, Salman Beigi, Andrew Drucker, Bill Fefferman, and Peter Shor.
\newblock The power of unentanglement.
\newblock {\em Theory of Computing}, 5(1):1--42, 2009,
  \href{http://arxiv.org/abs/0804.0802}
  {\textsc{arXiv:\linebreak[0]0804.0802}}.

\bibitem[AN02]{Aharonov2002}
Dorit Aharonov and Tomer Naveh.
\newblock Quantum {NP} - a survey.
\newblock October 2002, \href{http://arxiv.org/abs/quant-ph/0210077}
  {\textsc{arXiv:\linebreak[0]quant-ph/0210077}}.

\bibitem[Bab85]{Babai1985}
L\'{a}szl\'{o} Babai.
\newblock Trading group theory for randomness.
\newblock In {\em Proceedings of the 17th annual ACM Symposium on Theory of
  Computing}, STOC '85, pages 421--429, 1985.

\bibitem[BBD{\etalchar{+}}97]{Barenco1997}
Adriano Barenco, Andr\'{e} Berthiaume, David Deutsch, Artur Ekert, Richard
  Jozsa, and Chiara Macchiavello.
\newblock Stabilization of quantum computations by symmetrization.
\newblock {\em SIAM Journal on Computing}, 26(5):1541--1557, 1997,
  \href{http://arxiv.org/abs/quant-ph/9604028}
  {\textsc{arXiv:\linebreak[0]quant-ph/\linebreak[0]9604028}}.

\bibitem[BCWdW01]{Buhrman2001}
Harry Buhrman, Richard Cleve, John Watrous, and Ronald de~Wolf.
\newblock Quantum fingerprinting.
\newblock {\em Physical Review Letters}, 87(16):167902, September 2001,
  \href{http://arxiv.org/abs/quant-ph/0102001}
  {\textsc{arXiv:\linebreak[0]quant-ph/0102001}}.

\bibitem[Boo12]{Bookatz2012}
Adam~D. Bookatz.
\newblock {QMA}-complete problems.
\newblock December 2012, \href{http://arxiv.org/abs/1212.6312}
  {\textsc{arXiv:\linebreak[0]1212.\linebreak[0]6312}}.

\bibitem[Bra06]{Bravyi2006}
Sergey Bravyi.
\newblock Efficient algorithm for a quantum analogue of 2-{SAT}.
\newblock February 2006, \href{http://arxiv.org/abs/quant-ph/0602108}
  {\textsc{arXiv:\linebreak[0]quant-ph/0602108}}.

\bibitem[BSW11]{Beigi2011}
Salman Beigi, Peter Shor, and John Watrous.
\newblock Quantum interactive proofs with short messages.
\newblock {\em Theory of Computing}, 7(1):101--117, 2011,
  \href{http://arxiv.org/abs/1004.0411}
  {\textsc{arXiv:\linebreak[0]1004.0411}}.

\bibitem[BT09]{Blier2009}
Hugue Blier and Alain Tapp.
\newblock All languages in {NP} have very short quantum proofs.
\newblock In {\em Third International Conference on Quantum, Nano and Micro
  Technologies}, pages 34--37, 2009, \href{http://arxiv.org/abs/0709.0738}
  {\textsc{arXiv:\linebreak[0]0709.0738}}.

\bibitem[CKMR07]{Christandl2007}
Matthias Christandl, Robert K\"{o}nig, Graeme Mitchison, and Renato Renner.
\newblock One-and-a-half quantum de {F}inetti theorems.
\newblock {\em Communications in Mathematical Physics}, 273(2):473--498, 2007,
  \href{http://arxiv.org/abs/quant-ph/0602130}
  {\textsc{arXiv:\linebreak[0]quant-ph/0602130}}.

\bibitem[GMR89]{Goldwasser1989}
Shafi Goldwasser, Silvio Micali, and Charles Rackoff.
\newblock The knowledge complexity of interactive proof systems.
\newblock {\em SIAM Journal on Computing}, 18(1):186--208, 1989.

\bibitem[GN13]{Gosset2013}
David Gosset and Daniel Nagaj.
\newblock Quantum {3-SAT} is {QMA$_1$}-complete.
\newblock February 2013, \href{http://arxiv.org/abs/1302.0290}
  {\textsc{arXiv:\linebreak[0]1302.0290}}.

\bibitem[GS86]{Goldwasser1986}
Shafi Goldwasser and Michael Sipser.
\newblock Private coins versus public coins in interactive proof systems.
\newblock In {\em Proceedings of the 18th annual ACM Symposium on Theory of
  Computing}, STOC '86, pages 59--68, 1986.

\bibitem[HM10]{Harrow2010}
Aram~W. Harrow and Ashley Montanaro.
\newblock An efficient test for product states with applications to quantum
  {M}erlin-{A}rthur games.
\newblock In {\em 51st Annual IEEE Symposium on Foundations of Computer
  Science}, pages 633--642, 2010, \href{http://arxiv.org/abs/1001.0017}
  {\textsc{arXiv:\linebreak[0]1001.0017}}.

\bibitem[JJUW10]{Jain2010}
Rahul Jain, Zhengfeng Ji, Sarvagya Upadhyay, and John Watrous.
\newblock {QIP} = {PSPACE}.
\newblock In {\em Proceedings of the 42nd annual ACM Symposium on Theory of
  Computing}, STOC '10, pages 573--582, 2010,
  \href{http://arxiv.org/abs/0907.4737}
  {\textsc{arXiv:\linebreak[0]0907.4737}}.

\bibitem[JKNN12]{Jordan2012}
Stephen~P. Jordan, Hirotada Kobayashi, Daniel Nagaj, and Harumichi Nishimura.
\newblock Achieving perfect completeness in classical-witness quantum
  {M}erlin-{A}rthur proof systems.
\newblock {\em Quantum Information and Computation}, 12(5--6):461--471, May
  2012, \href{http://arxiv.org/abs/1111.5306}
  {\textsc{arXiv:\linebreak[0]1111.5306}}.

\bibitem[JN12]{Jain2012a}
Rahul Jain and Ashwin Nayak.
\newblock Short proofs of the quantum substate theorem.
\newblock {\em IEEE Transactions on Information Theory}, 58(6):3664--3669, June
  2012, \href{http://arxiv.org/abs/1103.6067}
  {\textsc{arXiv:\linebreak[0]1103.6067}}.

\bibitem[KKMV08]{Kempe2008}
Julia Kempe, Hirotada Kobayashi, Keiji Matsumoto, and Thomas Vidick.
\newblock Using entanglement in quantum multi-prover interactive proofs.
\newblock In {\em 23rd Annual IEEE Conference on Computational Complexity},
  pages 211--222, June 2008, \href{http://arxiv.org/abs/0711.3715}
  {\textsc{arXiv:\linebreak[0]0711.3715}}.

\bibitem[KKR06]{Kempe2006}
Julia Kempe, Alexei Kitaev, and Oded Regev.
\newblock The complexity of the local hamiltonian problem.
\newblock {\em SIAM Journal on Computing}, 35(5):1070--1097, 2006,
  \href{http://arxiv.org/abs/quant-ph/0406180}
  {\textsc{arXiv:\linebreak[0]quant-ph/0406180}}.

\bibitem[KLGN13]{Kobayashi2013}
Hirotada Kobayashi, Fran\c{c}ois Le~Gall, and Harumichi Nishimura.
\newblock Stronger methods of making quantum interactive proofs perfectly
  complete.
\newblock In {\em Proceedings of the 4th conference on Innovations in
  Theoretical Computer Science}, ITCS '13, pages 329--352, New York, NY, USA,
  2013. ACM, \href{http://arxiv.org/abs/1210.1290}
  {\textsc{arXiv:\linebreak[0]1210.1290}}.

\bibitem[KMY03]{Kobayashi2003}
Hirotada Kobayashi, Keiji Matsumoto, and Tomoyuki Yamakami.
\newblock Quantum {M}erlin-{A}rthur proof systems: {A}re multiple {M}erlins
  more helpful to {A}rthur?
\newblock In {\em Algorithms and Computation}, volume 2906 of {\em Lecture
  Notes in Computer Science}, pages 189--198. Springer Berlin / Heidelberg,
  2003, \href{http://arxiv.org/abs/quant-ph/0306051}
  {\textsc{arXiv:\linebreak[0]quant-ph/\linebreak[0]0306051}}.

\bibitem[Kni96]{Knill1996}
Emanuel Knill.
\newblock Quantum randomness and nondeterminism.
\newblock October 1996, \href{http://arxiv.org/abs/quant-ph/9610012}
  {\textsc{arXiv:\linebreak[0]quant-ph/9610012}}.

\bibitem[KSV02]{Kitaev2002}
A.~Yu. Kitaev, A.~H. Shen, and M.~N. Vyalyi.
\newblock {\em Classical and Quantum Computation}.
\newblock American Mathematical Society, 2002.

\bibitem[KW00]{Kitaev2000}
Alexei Kitaev and John Watrous.
\newblock Parallelization, amplification, and exponential time simulation of
  quantum interactive proof systems.
\newblock In {\em Proceedings of the 32nd annual ACM Symposium on Theory of
  Computing}, STOC '00, pages 608--617, 2000.

\bibitem[LFKN92]{Lund1992}
Carsten Lund, Lance Fortnow, Howard Karloff, and Noam Nisan.
\newblock Algebraic methods for interactive proof systems.
\newblock {\em Journal of the ACM}, 39(4):859--868, October 1992.

\bibitem[MW05]{Marriott2005}
Chris Marriott and John Watrous.
\newblock Quantum {A}rthur-{M}erlin games.
\newblock {\em Computational Complexity}, 14(2):122--152, 2005,
  \href{http://arxiv.org/abs/cs/0506068}
  {\textsc{arXiv:\linebreak[0]cs/0506068}}.

\bibitem[NC00]{Nielsen2000}
Michael~A. Nielsen and Isaac~L. Chuang.
\newblock {\em {Quantum Computation and Quantum Information}}.
\newblock Cambridge University Press, 2000.

\bibitem[NWZ09]{Nagaj2009}
Daniel Nagaj, Pawel Wocjan, and Yong Zhang.
\newblock Fast amplification of {QMA}.
\newblock {\em Quantum Information and Computation}, 9(11):1053--1068, November
  2009, \href{http://arxiv.org/abs/0904.1549}
  {\textsc{arXiv:\linebreak[0]0904.1549}}.

\bibitem[Per12]{Pereszlenyi2012}
Attila Pereszl\'{e}nyi.
\newblock Multi-prover quantum {M}erlin-{A}rthur proof systems with small gap.
\newblock May 2012, \href{http://arxiv.org/abs/1205.2761}
  {\textsc{arXiv:\linebreak[0]1205.2761}}.

\bibitem[Sha92]{Shamir1992}
Adi Shamir.
\newblock {IP} = {PSPACE}.
\newblock {\em Journal of the ACM}, 39(4):869--877, October 1992.

\bibitem[She92]{Shen1992}
A.~Shen.
\newblock {IP} = {PSPACE}: Simplified proof.
\newblock {\em Journal of the ACM}, 39(4):878--880, October 1992.

\bibitem[Vya03]{Vyalyi2003}
Mikhail~N. Vyalyi.
\newblock {QMA} = {PP} implies that {PP} contains {PH}.
\newblock Technical Report 21 (2003), Electronic Colloquium on Computational
  Complexity, April 2003.
\newblock \href{http://eccc.hpi-web.de/report/2003/021/}{TR03-021}.

\bibitem[Wat00]{Watrous2000}
John Watrous.
\newblock Succinct quantum proofs for properties of finite groups.
\newblock In {\em 41st Annual IEEE Symposium on Foundations of Computer
  Science}, pages 537--546, 2000, \href{http://arxiv.org/abs/cs/0009002}
  {\textsc{arXiv:\linebreak[0]cs/0009002}}.

\bibitem[Wat03]{Watrous2003}
John Watrous.
\newblock {PSPACE} has constant-round quantum interactive proof systems.
\newblock {\em Theoretical Computer Science}, 292(3):575--588, 2003.

\bibitem[Wat08a]{Watrous2008a}
John Watrous.
\newblock Quantum computational complexity.
\newblock April 2008, \href{http://arxiv.org/abs/0804.3401}
  {\textsc{arXiv:\linebreak[0]0804.\linebreak[0]3401}}.

\bibitem[Wat08b]{Watrous2008}
John Watrous.
\newblock Theory of quantum information.
\newblock Lecture notes from Fall 2008,
  \url{https://cs.uwaterloo.ca/~watrous/quant-info/}, 2008.

\bibitem[Wat09]{Watrous2009}
John Watrous.
\newblock Zero-knowledge against quantum attacks.
\newblock {\em SIAM Journal on Computing}, 39(1):25--58, 2009,
  \href{http://arxiv.org/abs/quant-ph/0511020}
  {\textsc{arXiv:\linebreak[0]quant-ph/0511020}}.

\bibitem[Win99]{Winter1999}
Andreas Winter.
\newblock {\em Coding Theorems of Quantum Information Theory}.
\newblock PhD thesis, Universit\"{a}t Bielefeld, 1999,
  \href{http://arxiv.org/abs/quant-ph/9907077}
  {\textsc{arXiv:\linebreak[0]quant-ph/9907077}}.

\bibitem[ZF87]{Zachos1987}
Stathis Zachos and Martin F\"{u}rer.
\newblock Probabalistic quantifiers vs. distrustful adversaries.
\newblock In {\em Proceedings of the Seventh Conference on Foundations of
  Software Technology and Theoretical Computer Science}, volume 287 of {\em
  Lecture Notes in Computer Science}, pages 443--455, London, UK, 1987.
  Springer-Verlag.

\end{thebibliography}
\end{document}